\documentclass[%
 aip,
 amsmath,amssymb,
 reprint,%
]{revtex4-1}

\usepackage{graphicx}
\usepackage{dcolumn}
\usepackage{bm}

\usepackage[utf8]{inputenc}
\usepackage[T1]{fontenc}
\usepackage{mathptmx}
\usepackage{amsmath,amsfonts,amssymb,amsthm,amscd}
\usepackage{graphicx,mathtools}
\usepackage{perpage}
\usepackage{url}
\usepackage{color}
\usepackage[T1]{fontenc}
\usepackage{subcaption}
\usepackage{mathrsfs}
\usepackage{booktabs}

\newtheorem{theorem}{Theorem}[section]
\newtheorem{lemma}[theorem]{Lemma}
\newtheorem{proposition}[theorem]{Proposition}

\newtheorem{remark}[theorem]{Remark}
\newtheorem{definition}[theorem]{Definition}

\DeclareMathOperator{\sync}{sync}
\DeclareMathOperator{\psync}{par-sync}

\DeclareMathOperator{\sym}{sym}
\DeclareMathOperator{\zero}{zero}
\DeclareMathOperator{\pu}{pop-up}
\DeclareMathOperator{\pd}{pop-down}

\newcommand{\R}{\mathbb{R}}

\renewcommand{\S}{\mathbb{S}}
\newcommand{\RR}{\mathbf{\mathscr{R}}}

\newcommand{\eee}{\mathrm{e}}
\newcommand{\ee}{\varepsilon}
\newcommand{\ddd}{\mathrm{d}}
\newcommand{\dd}{\delta}

\newcommand{\mi}{\mathrm{i}}

\newcommand{\p}{\partial}

\begin{document}

\preprint{AIP/123-QED}

\title[Two-community noisy Kuramoto model with general interaction strengths: Part II]{Two-community noisy Kuramoto model with general interaction strengths: Part II}

\author{S. Achterhof}
\affiliation{ 
Mathematical Institute, Leiden University, P.O.\ Box 9512,
2300 RA Leiden, The Netherlands.
}%

\author{J. M. Meylahn}
 \email{j.m.meylahn@uva.nl}
\affiliation{%
Amsterdam Business School, University of Amsterdam, P.O.\ Box 15953,
1001 NL Amsterdam, The Netherlands.
}%

\date{\today}

\begin{abstract}
We generalize the study of the noisy Kuramoto model, considered on a network of two interacting communities, to the case where the interaction strengths within and across communities are taken to be different in general. Using a geometric interpretation of the self-consistency equations developed in Part I of this series as well as perturbation arguments we are able to identify all solution boundaries in the phase diagram. This allows us to completely classify the phase diagram in the four dimensional parameter space and identify all possible bifurcation points. Furthermore, we analyze the asymptotic behavior of the solution boundaries. To illustrate these results and the rich behavior of the model we present phase diagrams for selected regions of the parameter space.
\end{abstract}

\maketitle

\begin{quotation}
The two-community noisy Kuramoto model is used to study synchronization on two communities of oscillators, interacting within and across the communities. The two community structure is relevant for neurophysiologists (e.g., to describe the body clock\cite{Rohling2020}) and social scientist (e.g., to analyze polarized opinion formation\cite{Hong2011a, Pluchino2005, Xiao2019}). The stationary states (the states after waiting a long time) of the system solve a system of equations that cannot be solved analytically. We analyze where phase transitions occur, i.e., where new stationary states occur when varying the parameters of the model. 
\end{quotation}

\section{Motivation and Background}

The mean-field Kuramoto model introduced by Kuramoto in $1975$ \cite{K75} has been extensively studied in the literature with a notable recent contribution covering aspects like phase transitions, the effect of disorder and stability in the noisy variant of the model \cite{L12a}. The model captures the phenomenon of synchronization, which is omnipresent in nature. For example,  fireflies flash when isolated at their own natural frequency but adapt their flashing to the rhythm of the other fireflies when in a group. The globally synchronized state arising here is due to local interactions and not a central driving mechanism. Whether a phase transition to the synchronized state occurs or not depends on the strength of the interactions. Once the interaction strength exceeds a critical value the system will reach a stable, synchronized state \cite{S88}.

We extend this fundamental result to the noisy Kuramoto model on a two-community network where we have four general interaction strength (two internal interaction strengths and two external interaction strengths). In this extension of the mean-field noisy Kuramoto model we investigate the phase transitions when there is no disorder.

In the previous paper\cite{Achterhof2020} of this series we showed that the average phase between the communities is either zero or $\pi$, significantly simplifying the analysis. Furthermore, using a geometric interpretation of the self-consistency equations we split the phase space into ten regions and derived upper bounds on the number of solutions in each region. In this paper we refine this result by characterizing the full phase diagram, i.e., we identify all solution boundaries (phase transitions) and establish the number of (stationary) synchronized solutions in the resulting regions.

The new results on the two-community noisy Kuramoto model are relevant for neurophysiologists, since they may explain some phenomena observed in the functioning of the suprachiasmatic nucleus (SCN). The SCN, or body clock, is a two community network (of approximately $10^{4}$ neurons per community in humans) in the brain of mammals responsible for dictating most bodily rhythms. The mathematical results of the symmetric model in \cite{Meylahn2020} could explain, for example, the transition of the SCN to a phase-split state in certain light conditions \cite{Rohling2020}. Experiments show that the presence or absence of chemicals in the SCN changes the strength of interaction between the neurons. This so called E/I balance (excitatory and inhibitory balance) is influenced by environmental factors, for example, the exposure to light. From this we conclude that the interaction strengths are time dependent and the pairs of internal and external interaction strengths are not necessarily equal. The results presented in this paper are a first step to further the understanding of the mechanics of the SCN.  

The paper proceeds as follows. In Section \ref{sec:model} we state the relevant results of Part I \cite{Achterhof2020}. In Section \ref{sec:bifurcation} we partition the ten fundamental regions of Part I into subregions by identifying solution boundaries using a perturbation argument. We also compute the asymptotic behavior of these solution boundaries. In Section \ref{sec:classification} we use the solution boundaries to describe the phase diagram in the ten fundamental regions. Finally, in Section \ref{sec:phasediagrams} we give numerical examples of the phase diagrams in specific regions.

\label{sec:intro}
\section{Model}
\label{sec:model}
Consider two populations of oscillators, both of size $N$, with internal mean-field interactions of strength $\frac{K_1}{N}$ and $\frac{K_2}{N}$. Furthermore, the oscillators in community $1$ experience a mean-field interaction with the oscillators in community 2 of strength $\frac{L_1}{N}$ and the oscillators in community $2$ experience a mean-field interaction with the oscillators of community 1 of strength $\frac{L_2}{N}$. We assume that $K_1, K_2 \in \R$ and $L_1, L_2 \in \R \setminus\{0\}$.
\begin{definition}[Two-community noisy Kuramoto model]
The evolution of $\theta_{1,i}$, $i = 1, \ldots, N$, on $\S = \R/ 2\pi$ is governed by the SDE
\begin{align}
\ddd \theta_{1,i}(t) &= \frac{K_1}{2N} \sum_{k = 1}^{N} \sin( \theta_{1,k}(t) - \theta_{1,i}(t) ) \ddd t  \label{kura1} \\
 	&+ \frac{L_1}{2N} \sum_{l = 1}^{N} \sin(\theta_{2,l}(t) - \theta_{1,i}(t) ) \ddd t + \ddd W_{1,i}(t).\nonumber  
\end{align}

As initial condition we take $\theta_{1,i}(0), i = 1, \ldots, N,$ which are i.i.d. and are drawn from a common probability distribution $\rho_1$ on $\S$.
\\

The phase angles of the oscillators in community 2 are denoted by $\theta_{2,j}$, $j = 1, \ldots, N$, and their evolution on $\S = \R/ 2\pi$ is governed by the SDE
\begin{align}
\ddd \theta_{2,j}(t) &=  \frac{K_2}{2N} \sum_{l = 1}^{N} \sin( \theta_{2,l}(t) - \theta_{2,j}(t) ) \ddd t \nonumber \\
 	&+ \frac{L_2}{2N} \sum_{k = 1}^{N} \sin(\theta_{1,k}(t) - \theta_{2,j}(t) ) \ddd t + \ddd W_{2,j}(t). \label{kura2}
\end{align}
As initial condition we take $\theta_{2,j}(0), j = 1, \ldots, N_2,$ are i.i.d. drawn from a common probability distribution $\rho_2$ on $\S$. 
Furthermore $\left(W_{1,i}\right)_{t \geq 0}$, $i = 1,\ldots, N$ and $\left(W_{2,j}\right)_{t \geq 0}$, $j = 1,\ldots, N$ are two independent standard Brownian motions.
\end{definition}
The order parameters, $r_{k}$ and $\psi_{k}$, are defined by
\begin{equation}
r_k \eee^{\mi \psi_k} := \int_\S \eee^{\mi \theta } p_k(\theta) \ddd \theta , \label{nop1}
\end{equation}
with $p_k(\theta)$ the steady-state distribution of the oscillators. We refer to $r_{k}$ as the synchronization level of community $k$ and to $\psi_{k}$ as the average angle in community $k$. In the first part of this series of papers \cite{Achterhof2020} we showed that steady state solutions of the dynamics described by \eqref{kura1} and \eqref{kura2} in the limit as $N\rightarrow \infty$ must satisfy
\begin{align}
r_k	&=  V(K_k r_k + L_k r_{k'} \cos \psi), \label{eq:scp1}
\end{align}
for $k\in \{1, 2\}$. Here $k'$ denotes the complement of $k$, $\psi = \psi_{1}-\psi_{2}$ and
\begin{equation}
V(x) := \frac{\int_{\S}\cos \theta\eee^{x \cos \theta}\ddd\theta}{\int_{\S}\eee^{x \cos \theta}\ddd\theta}.
\end{equation}
Furthermore, we showed that in the steady state $\psi \in \{0, \pi\}$ and since any analysis of the self-consistency equation \eqref{eq:scp1} with $\psi=\pi$ is the same as the analysis with $\psi=0$ and $L_{k}\rightarrow -L_{k}$ we restrict ourselves to the case $\psi = 0$. 

The self-consistency equations \eqref{eq:scp1} are generally difficult to solve. In order to determine how many solutions are possible given the parameter values $K_{k}$, and $L_{k}$ we introduced a geometric interpretation of \eqref{eq:scp1} by defining the following curves.
\begin{definition}[Self-consistency intersection curves]
\label{def:sccurves}
\begin{align}
\Gamma_1^{K_1, L_1} &:= \left\{ (r_1, r_2) \in [0,1]^2 : h_1^{K_1, L_1}(r_1, r_2) = 0 \right\},\\
\Gamma_2^{K_2, L_2} &:= \left\{ (r_1, r_2) \in [0,1]^2 : h_2^{K_2, L_2}(r_1, r_2) = 0 \right\}.
\end{align}
with
\begin{align}
h_1^{K_1, L_1}(r_1, r_2) &:= V(K_1 r_1 + L_1  r_2) - r_1,\\
h_2^{K_2, L_2}(r_1, r_2) &:= V(K_2 r_2 + L_2  r_1) - r_2, 
\end{align}
\end{definition} 
A solution to \eqref{eq:scp1} corresponds to an intersection of the curves defined above. The self-consistency curve $\Gamma_{1}^{K_{1}, L_{1}}$ in Definition \ref{def:sccurves} can fall into one of the following three categories: `Convex curve connected with zero' (corresponding to $K_{1}\leq 2, L_{1}>0$), `Convex curve disconnected from zero' (corresponding to $K_{1}>2, L_{1}>0$) and `Parabola' (corresponding to $K_{1}>2, L_{1}<0$). A given set of parameters falls within one of the nine regions corresponding to the nine combinations of curves based that are possible. Using the properties of the curves in the given domain allows us to determine the maximum number of solutions to \eqref{eq:scp1} that are possible for that set of parameters. 

This leads to the classification of the phase space given in Table \ref{fig:overview} taken from \cite{Achterhof2020}.
\begin{table}[!ht]
\centering
\vspace{0.1cm}

\begin{tabular}{l l c} \toprule
		& \textbf{Region}  &\textbf{ Max $\#$ solutions}  \\ \midrule

$\RR_1$	& $K_1 < 2, L_1 < 0 \text{ or } K_2 < 2, L_2  < 0$	&  $1$\\ \hline
$\RR_2$	& $K_1 \leq 2, K_2 \leq 2, L_1 > 0, L_2  > 0$	&  $2$\\ 

$\RR_3$	& $K_1 >  2, K_2 > 2, L_1 > 0, L_2  > 0$		&  $2$\\ 

$\RR_4$	& $K_1 \leq 2, K_2 > 2, L_1 > 0, L_2 > 0$		&  $2$\\ 
$\RR_5$	& $K_1 > 2, K_2 \leq 2, L_1  > 0, L_2  > 0$		&  $2$\\ \hline	
	
$\RR_6$	& $K_1>2, K_2>2, L_1 < 0, L_2> 0$		&  $3$\\ 		
$\RR_7$	& $K_1>2, K_2>2, L_1 > 0, L_2 < 0$		&  $3$\\ 
			
$\RR_8$ & $K_1>2, K_2 \leq 2, L_1  < 0, L_2  > 0$	&  $3$\\ 
$\RR_9$	& $K_1 \leq 2, K_2>2, L_1  > 0, L_2  < 0$	&  $3$\\ \hline

$\RR_{10}$	& $K_1>2, K_2 > 2, L_1 < 0, L_2 < 0$		&  $4$\\ \bottomrule
\end{tabular}
\caption{An overview of all regions in which synchronized solutions can occur and the maximum number of solutions possible (unsycnhronized solutions included)}\label{fig:overview}
\end{table} 

\section{Solution boundaries}
\label{sec:bifurcation}
In this section we develop a method to further refine the regions identified in Table \ref{fig:overview}. Our goal is to partition the existing regions into sub-regions for which we know how many solutions occur. We call the edges of these sub-regions the \emph{solution boundaries}, which are curves $\beta = 0$ with $\beta: \R^4 \to \R$ a scalar function. At a solution boundary, new (synchronized) solutions occur or disappear. We develop a method to distinguish the solution boundaries. If we fix three of the four interaction strengths, e.g. $K_2, L_1, L_2$, then a \emph{solution orbit} $K_1 \mapsto (r_1(K_1, r_2(K_1))$ defines a dynamical system. In this context, a \emph{bifurcation} occurs when a small change of one of the interaction strengths $(K_1, K_2, L_1, L_2)$ causes the appearance of a new solution orbit. For example, in  Figure \ref{fig:bzero1} we see that a solution bifurcates continuously from the unsynchronized solution. Clearly, at a solution boundary bifurcation occurs. Hence, in order to determine the solution boundaries, we need to determine where bifurcation occurs.\\

A parameter value at which a bifurcation occurs is called a \emph{bifurcation point}. Note that, because of the dependence of the bifurcation point on all other interaction strength, we have without loss of generality that the bifurcation point is of the form $K_1(K_2,L_1,L_2)$. A synchronization level $(r_1, r_2)$, $r_1, r_2 \in (0,1)$ from which a new solution orbit splits off is called a \emph{bifurcation level}. In this section we consider three types of bifurcation:
\begin{enumerate}
\item Section \ref{sec:bifsyn}: Bifurcation from a synchronized solution $(r_1, r_2)$ for $r_1, r_2 \in (0,1)$. We denote by $\beta^{\sync} = 0$ the corresponding solution boundary.
\item Section \ref{sec:bif0}: Bifurcation from the unsynchronized solution $(r_1, r_2) = (0,0)$. We denote by $\beta^{\zero} = 0$ the corresponding solution boundary.
\item Section \ref{sec:biflim}: Bifurcation from a partial synchronized solution $(0,r_2)$ or $(r_1,0)$ for $r_1, r_2 \in (0,1)$. We denote by $\beta^{\psync} = 0$ the corresponding solution boundary. 
\end{enumerate}

\subsection{Bifurcation from a synchronized solution}
\label{sec:bifsyn}

\begin{theorem}
\label{thm:biffromsync}
The solution boundary (if it exists) at which a bifurcation from a synchronized solution appears is given by 
\begin{align}
\beta^{\sync}(K_1, K_2, L_1, L_2) :=&  (L_1 L_2 - K_1 K_2) C_{1,1} C_{2,1}\nonumber \\
&+ K_1 C_{1,1} + K_2 C_{2,1} - 1 = 0, \label{eq:eta=0}
\end{align}
for some synchronized solution $(r_1, r_2) \in \Gamma^{K_1, L_1}_1 \cap \Gamma^{K_2, L_2}_2$, where $C_{1,1}, C_{1,2}, C_{2,1}$ and $C_{2,2}$ are defined as
\begin{align}
\label{eq:Cderivatives1}
C_{1,1} &= V'(K_1 r_1 + L_1 r_2),\quad C_{1,2} = V''(K_1 r_1 + L_1 r_2),\\
C_{2,1} &= V'(K_2 r_2 + L_2 r_1),\quad C_{2,2} = V''(K_2 r_2 + L_2 r_1).
\label{eq:Cderivatives2}
\end{align}
\end{theorem}

\begin{proof} Assume that a new synchronized solutions bifurcates from a synchronized solution $(r_1, r_2) \neq (0,0)$. Then we can perform the following perturbation:  
\begin{align}
r_1 + \ee &= V( K_1 (r_1 + \ee) + L_1 (r_2 - \dd) ), \label{eq:per1}\\
r_2 - \dd &= V( K_2 (r_2 - \dd) + L_2 (r_1 + \ee) ). \label{eq:per2}
\end{align}
A Taylor expansion of $V(K_1(r_1 + \ee) + L_1(r_2 - \dd) )$  around the point $K_1 r_1 + L_2 r_2$ gives
\begin{align}
V(K_1(r_1 + \ee) + L_1(r_2 - \dd) ) =& (K_1 \ee - L_1 \dd ) V'(K_1 r_1 + L_1 r_2 ) \nonumber\\
&+ r_{1} + O( (\ee + \dd)^2 ).\label{eq:per3}
\end{align}
Combining \eqref{eq:per1} and \eqref{eq:per3} we get
\begin{equation}
\ee = (K_1 \ee - L_1  \dd ) V'((K_1 r_1 + L_1 r_2) + O( (\ee + \dd)^2 ). \label{eq:per4}
\end{equation}
Similarly, a Taylor expansion around $(K_2 r_2 + L_2 r_1)$ gives
\begin{equation}
- \dd = (L_2 \ee - K_2 \dd ) V'(K_2 r_2 + L_2 r_1) + O( (\ee + \dd)^2  ). \label{eq:per5}
\end{equation}
Combining \eqref{eq:per4} and \eqref{eq:per5}, we get
\begin{equation}
\ee \sim (K_1 \ee - L_1 \dd) C_{1,1}, \quad \text{and }\quad \dd \sim (K_2 \dd - L_2 \ee) C_{2,1}, \label{eq:asy1}
\end{equation}
as $\ee, \dd \downarrow 0$  with $C_{1,1}$ and $C_{2, 1}$ defined in \eqref{eq:Cderivatives1} and \eqref{eq:Cderivatives2}. Rewriting \eqref{eq:asy1}, we obtain
\begin{equation}
\ee \sim \frac{- L_1  C_{1,1}}{1 - K_1 C_{1,1}} \dd, \quad \text{and }\quad \dd \sim \frac{- L_2 C_{2,1}}{1 - K_2 C_{2,1}} \ee.
\label{eq:asy2}
\end{equation}
Combining the equations in \eqref{eq:asy2}, leads to
\begin{equation}
\ee \sim \frac{L_1 L_2 C_{1,1} C_{2,1}}{(1 - K_1 C_{1,1})(1 - K_2 C_{2,1})} \ee,
\end{equation}
which implies that 
\begin{equation}
\frac{L_1 L_2 C_{1,1} C_{2,1} }{ (1 - K_1 C_{1,1})(1 - K_2 C_{2,1})} = 1,
\end{equation}
from which the claim follows.
\end{proof}

\begin{remark}
The solution boundary $\beta^{\sync} = 0$ can be computed numerically by solving the following system of equations:
\begin{equation*}
\begin{cases}
r_1 &= V(K_1 r_1 + L_1 r_2),\\
r_2 &= V(K_2 r_2 + L_2 r_1),\\
0 	&= (L_1 L_2 - K_1 K_2) V'(K_1 r_1 + L_1 r_2) V'(K_2 r_2 + L_2 r_1)\\
	&+ ~ K_1 V'(K_1 r_1 + L_1 r_2) + K_2  V'(K_2 r_2 + L_2 r_1) - 1. 
\end{cases}
\end{equation*}
If we fix three of the four interaction strengths, then we have three equations with three unknowns.
\end{remark}

\subsection{Bifurcation from the unsynchronized solution}\label{sec:bif0}

Suppose that we fix $K_2, L_1, L_2$, and let $K_1$ vary. We are interested in finding the interaction strengths $K_1^{\zero}, K_2^{\zero}, L_1^{\zero}$ and $L_2^{\zero}$ where the orbits $K_1 \mapsto r_1(K_1)$ and $K_1 \mapsto r_2(K_1)$ split off from the unsynchronized solution, i.e., when they bifurcates from zero. Note that when bifurcation from zero occurs both orbits split off simultaneously from zero, because by \cite[Theorem II.7]{Achterhof2020} solution pairs $(r_1, r_2)$ with $r_1 > 0$ and $r_2 = 0$ (and vice-versa) do not exist.

\begin{figure}[!ht]
\begin{center}
\includegraphics[width=0.4\textwidth]{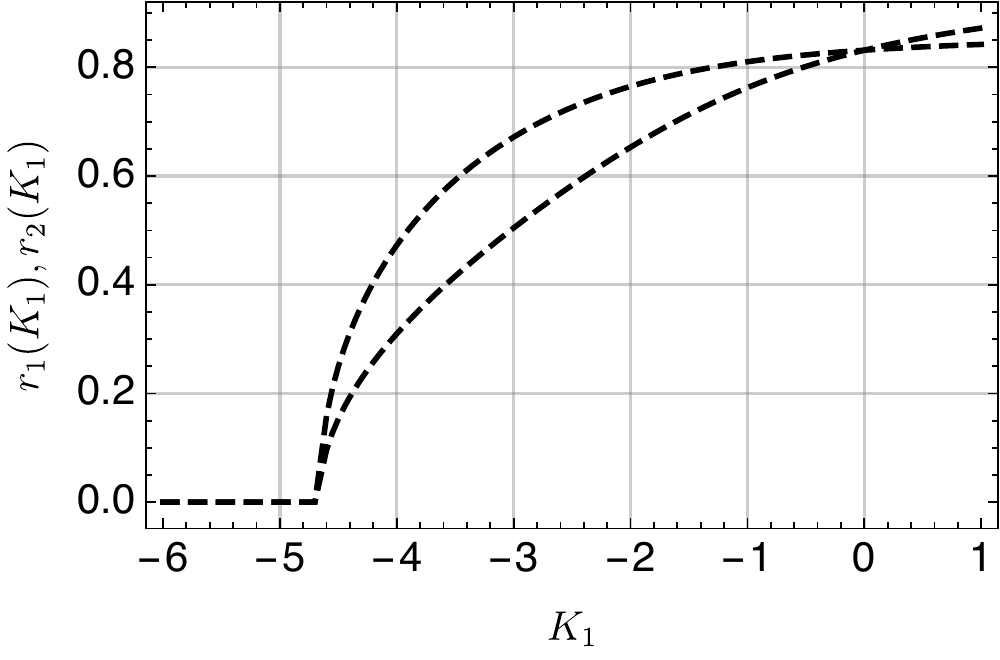} 
 \caption{Plot of $K_1 \mapsto r_1(K_1)$ and $K_2 \mapsto r_2(K_1)$, with $K_2 = -1, L_1 = 4$ and $L_2 = 5$. Bifurcation at zero occurs at $K_1^{\zero} = -\frac{14}{3}$, also note the symmetric solution at $K_1^{\sym} = 0$. }\label{fig:bzero1}
\end{center}
\end{figure}

\begin{definition}[Bifurcation from zero] A synchronized solution is said to \emph{bifurcate from zero} if it splits off continuously from the unsynchronized solution $(r_1, r_2) = (0,0)$ when one of the interaction strengths is varied (see Figure \ref{fig:bzero1}). We denote by $K_1^{\zero}, K_2^{\zero}, L_1^{\zero}$ and $L_2^{\zero}$ the interaction strengths corresponding with the solution that bifurcates from zero.
\end{definition}

\begin{lemma}\label{lem:nobif}
In the following cases bifurcation from zero is \emph{not} possible:
\begin{enumerate}
\item $K_1 < 2$ and $L_1 < 0$,
\item $K_2 < 2$ and $L_2 < 0$,
\item $K_1 > 2$ and $L_1 > 0$,
\item $K_2 > 2$ and $L_2 > 0$.
\end{enumerate}
\end{lemma}

\begin{proof}
The proof is geometrical. Once three of the interaction strength parameters are fixed, a variation of the last parameter will only change one of the two fundamental curves corresponding to the two self-consistency equations. A bifurcation at zero occurs when a change in the last parameter leads to a new intersection between these two fundamental curves at the point $(r_{1}, r_{2})=(0, 0)$. The properties of the fundamental curves allow us to exclude the possibility of this occurring in certain regions of the parameter space. 

(Case 1 + 2) In these cases the unsynchronized solution is the only solution, as they correspond to regions 7 and 8 of \cite[Theorem IV.1]{Achterhof2020} so that $\Gamma_1 = \Gamma_2 = \{(0,0)\}$. 

(Case 3 + 4) In these cases one of the two lines is always a "convex curve disconnected from zero", so that a new intersection cannot occur arbitrarily close to $(r_{1}, r_{2})=(0, 0)$.
\end{proof}

\begin{theorem}[Zero solution boundary]\label{thm:bifzero1}
Assume $(K_1, K_2, L_1, L_2)$ is not contained in the regions described in Lemma \ref{lem:nobif}. If $K_1 \neq 2$ and $K_2 \neq 2$, then bifurcation from zero occurs if and only if 
\begin{equation}
\beta^{\zero}(K_1,K_2, L_1, L_2) = (K_1 - 2)(K_2 - 2) - L_1 L_2 = 0.
\end{equation}
\end{theorem}

\begin{proof}
For the first implication we will perform a perturbation around zero, since the bifurcation we are considering will bifurcate continuously from zero. The calculation is a special case of the proof of Theorem \ref{thm:biffromsync} where the perturbation is around $r_{1}=0$ and $r_{2}=0$. For the reverse implication note that 
\begin{equation}
\beta^{\zero} = 0 \iff \frac{\p \Gamma_1}{\p r_1}(0,0) = \frac{\p \Gamma_2}{\p r_1}(0,0).
\end{equation}
In this geometric configuration a small change in one of the interaction strength leads to a bifurcation from the unsynchronized solution.
\end{proof}

Next we consider the cases where $K_1 = 2$ or $K_2 = 2$.

\begin{theorem}\label{thm:bifzero2}
Fix $K_1 = K_2 = 2$. Then bifurcation from zero occurs if and only if $\beta^{\zero}(2,2,L_1,L_2) = L_1 L_2 = 0$.
\end{theorem}

\begin{proof}
By \cite[Theorem IV.1]{Achterhof2020} we may assume that $L_1 > 0$ and $L_2 > 0$. In the case $K_1 = K_2 = 2$, with $r_{1}=0$ and $r_{2}=0$ due to perturbing around zero, equations \eqref{eq:per4} and \eqref{eq:per5} reduce to 
\begin{equation}
\ee  = \ee + \frac{1}{2} L_1 \dd + O((\ee + \dd)^2),
\end{equation}
and 
\begin{equation}
\dd = \dd + \frac{1}{2} L_2 \ee + O((\ee + \dd)^2).
\end{equation}
It follows that 
\begin{equation}
\ee \sim \ee + \frac{1}{2} L_1 \dd,\quad \text{and }\quad \dd \sim \dd + \frac{1}{2} L_2 \ee,
\end{equation}
which implies that 
\begin{equation}
\frac{L_1 \dd}{2 \ee} \to  0, \quad \frac{L_2 \ee}{2 \dd} \to  0, \quad \ee, \dd \downarrow 0.\label{eq:as1}
\end{equation}

From this we conclude that $\ee \not\sim \dd$. Now suppose that $\ee = o(\dd)$, i.e., $\ee$ is dominated by $\dd$. In this case the right-hand side of \eqref{eq:as1} is true for all $L_2 > 0$, but the left-hand side of \eqref{eq:as1} is only true when $L_1  = 0$. Similarly, if $\dd = o(\ee)$, we get that $L_2  = 0$.
\end{proof}

It remains to analyze the region $K_1 = 2$ and $K_2 \neq 2$ or $K_1 \neq 2$ and $K_2 = 2$. By Lemma \ref{lem:nobif} we can restrict our self to the case where $K_1 = 2$ and $K_2 < 2$ or $K_1 < 2$ and $K_2 = 2$. If $K_2 < 2$, then we must have $L_2  > 0$ to have a synchronized solution. Similarly, if $K_1 < 2$ we must have $L_1  > 0$. 

\begin{theorem}\label{thm:bifzero3}
A bifurcation at zero occurs:
\begin{enumerate}
\item if $K_1 = 2$, $K_2 \neq 2$, $L_1   > 0$ and $L_2 = 0$,
\item if $K_1 \neq 2$, $K_2 = 2$ $L_2 > 0$ and $L_1 = 0$.
\end{enumerate}
\end{theorem}

\begin{proof}
First assume $K_1 = 2$ and $K_2 < 2$. Then 
\begin{equation}
\ee \sim \ee + \frac{1}{2} L_1  \dd, \label{eq:as3}
\end{equation}
and
\begin{equation}
\delta \sim \frac{L_2 }{2 - K_2}\ee. \label{eq:as4}
\end{equation}
In order to have a synchronized solution, we require that $L_1  > 0$. Now for \eqref{eq:as3} to be true we require that $\dd = o(\ee)$. Note that \eqref{eq:as4} is true, if and only if $L_2  = 0$. The second case in the theorem follows from a similar argument.
\end{proof}

\subsection{Bifurcation from a partially synchronized solution}
\label{sec:biflim}
By Lemma \ref{lem:nobif} we know that bifurcation from zero is not possible when either $K_1 > 2$ and $L_1 > 0$ or $K_2 > 2$ and $L_2 > 0$. 

\begin{theorem}
\label{thm:biflim}
Bifurcation from a partially synchronized solution $(0,r)$, for some $r \in (0,1)$, occurs if and only if $L_1 = 0$ and $K_2 > 2$. Similarly, bifurcation from a partially synchronized solution $(r',0)$, for some $r' \in (0,1)$, occurs if and only if $L_2 = 0$ and $K_1 > 2$.
\end{theorem}
\begin{proof}
The perturbation we will perform here is a special case of the perturbation in the proof of Theorem \ref{thm:biffromsync} around $r_{1} = 0$ and $r_{2}=r$, which gives the solution boundary
\begin{equation}
L_1 L_2 V'(L_1 r) V'(K_2 r) - (1 - K_1 V'(L_1 r) ) (1 - K_2 V'(K_2 r) ) = 0. \label{eq:limboundary}
\end{equation}
This holds for some $r \in (0,1)$ that solves $V(L_1 r) = 0$, $V(K_2 r) = r$. The latter implies that $L_1 = 0$, $K_2 > 2$ and $V'(K_2 r) = \frac{1}{K_2}$. The case where we perturb from the limit point $(r',0)$, $r' \in (0,1)$ follows from a similar argument.
\end{proof}

\begin{definition}[Partial synchronized solution boundary]
Using Theorem \ref{thm:biflim} we can define the solution boundary $\beta^{\psync} = 0$ as follows:

\begin{equation}
\beta^{\psync}(K_1, K_2, L_1, L_2) := \begin{cases}
L_1, & \text{if } K_2 > 2,\\
L_2, & \text{if } K_1 > 2,\\
\emptyset,	& \text{else}.
\end{cases}
\end{equation}

\end{definition}

The kind of solutions that appear at the bifurcation point can be characterized as follows.

\begin{definition}
A \emph{pop-up solution} is a synchronized solution $(r_1^{\pu}, r_2^{\pu})$ that discontinuously appears as the parameter values are varied (see Figure \ref{fig:bpopl1} and Figure \ref{fig:bpopl2}). We denote by $K_1^{\pu}, K_2^{\pu}, L_1^{\pu}$ and $L_2^{\pu}$ the interaction strengths corresponding with the pop-up solution.
\end{definition}

\begin{figure}[!ht]
\begin{center}
\includegraphics[width=0.4\textwidth]{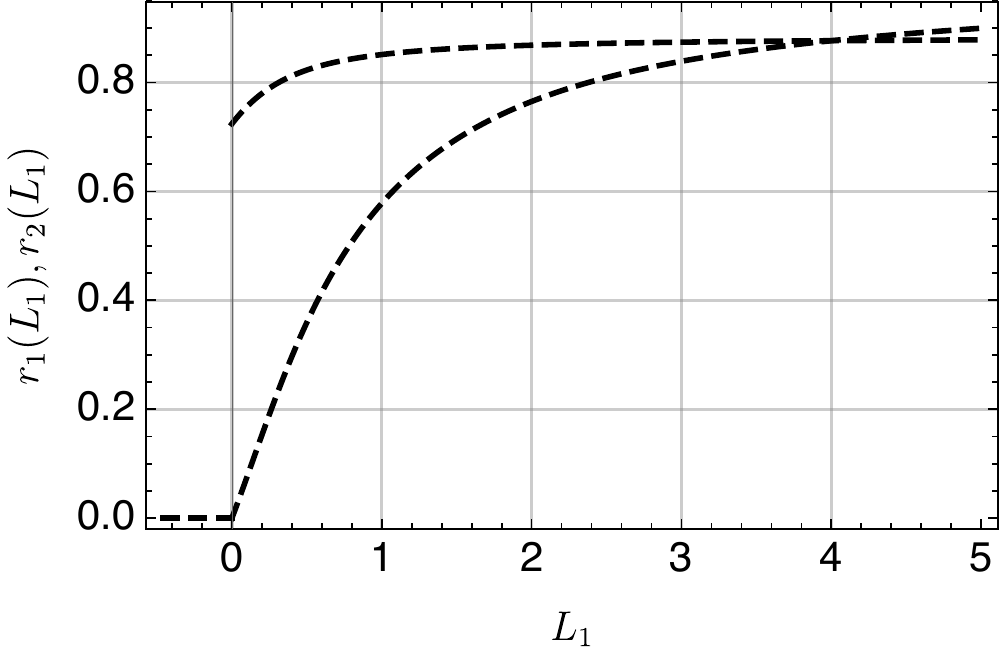} 
 \caption{Plot of $L_1 \mapsto r_1(L_1)$ and $K_2 \mapsto r_2(L_1)$, with $K_1 = 1, K_2 = 3$ and $L_2 = 2$. Pop-up bifurcation occurs at $L_1^{\pu} = 0$ and $(r_1^{\pu}, r_2^{\pu}) \approx (0,0.724159)$. In addition, $L_1^{\sym} = 4$. }\label{fig:bpopl1}
\end{center}
\end{figure}

\begin{theorem}[Pop-up bifurcation]\label{thm:epop}
Let $K_1 \leq 2$, $K_2 > 2$ and $L_2 > 0$. 
\begin{enumerate}
\item If $L_1 < 0$, then $r_1(L_1) = 0$ and $r_2(L_1) = 0$.
\item If $L_1 > 0$, then the non-trivial solution orbit $L_1 \mapsto (r_1(L_1), r_2(L_1))$ satisfies
\begin{equation}
\lim_{L_1 \downarrow 0} r_1(L_1) = 0,
\end{equation}
and there exists a unique $r \in (0,1)$ that solves $r = V(K_2 r)$ and
\begin{equation}
\lim_{L_1 \downarrow 0} r_2(L_1) = r.
\end{equation}
\end{enumerate}
\end{theorem}

\begin{proof}
The first part is a direct consequence of \cite[Theorem IV.1]{Achterhof2020}. For the second part: by continuity of $V$ we have that 
\begin{equation}
\lim_{L_1 \downarrow 0} r_1(L_1) = V\left(K_1 \left[ \lim_{L_1 \downarrow 0} r_1(L_1) \right] \right), \label{eq:vlim1}
\end{equation}
and 
\begin{equation}
\lim_{L_1 \downarrow 0} r_2(L_1) = V\left( \lim_{L_1 \downarrow 0} \left[ K_2 r_2(L_1) + L_2 r_1(L_1) \right] \right). \label{eq:vlim2}
\end{equation}
We have $K_1 \leq 2$ and therefore $\lim_{L_1 \downarrow 0} r_1(L_1) = 0$, which reduces (\ref{eq:vlim2}) to
\begin{equation}
\lim_{L_1 \downarrow 0} r_2(L_1) =  V\left(  K_2 \left[\lim_{L_1 \downarrow 0} r_2(L_1) \right] \right).
\end{equation}
By assumption $K_2 > 2$ and therefore $r = V(K_2 r)$ has an unique solution $r \in (0,1)$. The result now follows.
\end{proof}

\begin{remark}
In the region $K_1 >2$, $L_1 > 2$ and $K_2 \leq 2$, one can show by the same reasoning that if $L_2 < 0$, then $r_1(L_2) = 0$ and $r_2(L_2) = 0$. In addition, there exists a solution orbit $L_2 \mapsto (r_1(L_2), r_2(L_2))$ and an unique $r \in (0,1)$ which solves $r = V(K_2 r)$ and
\begin{equation}
\lim_{L_2 \downarrow 0} r_1(L_2) = r, \quad \lim_{L_2 \downarrow 0} r_2(L_2) = 0.
\end{equation}
\end{remark}

\subsection{Existence and asymptotics of the solution boundary}
We say that a solution boundary $\beta^{\sync} = 0$ exists in a set $A \subset \R^4$ if 
\begin{equation}
\{ \beta^{\sync}(K_1, K_2, L_1, L_2) = 0 : (K_1, K_2, L_1, L_2) \in A \} \neq \emptyset. \label{eq:sbex}
\end{equation}
In order to simplify the numerical computation of this boundary we need to know where the solution boundary $\{\beta^{\sync} = 0\}$ exists. To do so, we observe that:
\begin{equation}
(K_1, K_2, L_1, L_2) \in \{ \beta^{\sync} = 0 \} 
\end{equation}
if and only if
\begin{equation}
 \frac{\p \Gamma^{K_1, L_1}_1}{\p r_1}(\hat{r}_1, \hat{r}_2) = \frac{\p \Gamma^{K_2, L_2}_2}{\p r_1}(\hat{r}_1, \hat{r}_2),\label{eq:eqetadir}
\end{equation}
for some $(\hat{r}_1, \hat{r}_2) \in \Gamma_1^{K_1, L_1} \cap \Gamma_2^{K_2, L_2}$. Hence we can think geometrically about the level curve $\beta^{\sync} = 0$, i.e, we search for all $(K_1, K_2, L_1, L_2)$ and $(r_1, r_2) \in \Gamma_1^{K_1, L_1} \cap \Gamma_2^{K_2, L_2}$ where the derivatives $\p \Gamma^{K_1, L_1}_1/\p r_1$ and $\p \Gamma^{K_2, L_2}_2/\p r_1$ are equal. By the geometry of the level curves, this equality is not possible in the case that $L_1 > 0$ and $L_2 > 0$, because it is necessary that at least one of the level curves has a turning point (see Figure \ref{fig:eqdir} for examples).

\begin{lemma}[Region of non-existence]
\label{lem:solempty}
If $L_1 > 0$ and $L_2 > 0$, then $\{ \beta^{\sync} = 0 \} = \emptyset$.
\end{lemma}

\begin{proof}
This is clear from the discussion above.
\end{proof}

\begin{lemma}[Regions of existence]\label{lem:solnonempty}
If one of the three is true
\begin{enumerate}
\item $L_1 > 0$ and $L_2 < 0$,
\item $L_1 < 0$ and $L_2 > 0$,
\item $L_1 < 0$ and $L_2 < 0$,
\end{enumerate}
then 
\begin{equation}
\left\{ (K_1, K_2) \in \R^2 \setminus \{ (0,0) \} : \beta^{\sync}(K_1, K_2, L_1, L_2) = 0 \right\} \neq \emptyset.
\end{equation}
\end{lemma}

\begin{proof}
Suppose that $(L_1, L_2)$ is contained in one of the three regions described in Lemma \ref{lem:solnonempty}. Now  $\Gamma_1$ or $\Gamma_2$ has a turning point inside the unit square $[0,1]^2$ for a suitable choice of $K_1$ or $K_2$ (see Property 1 of \cite[Theorem III.11]{Achterhof2020}). Hence by the geometry of the level curves \eqref{eq:eqetadir} is satisfied for a suitable choice of $K_1$ or $K_2$.
\end{proof}

To fully describe the domain of existence of the solution boundary $\{\beta^{\sync} =  0\}$ we make a case distinction. In the first case $L_{1}$ and $L_{2}$ have opposite sign and in the second they are both negative. 

\begin{figure}[!ht]
\begin{center}
\includegraphics[width=0.4\textwidth]{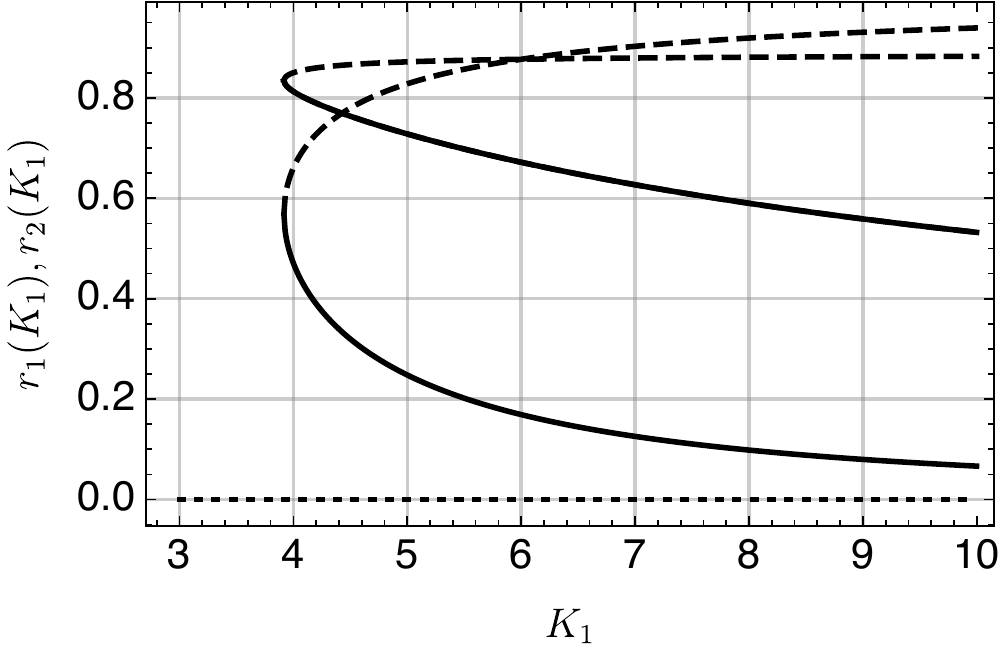} 
 \caption{Plot of $K_1 \mapsto r_1(K_1)$ and $K_2 \mapsto r_2(K_1)$, with $K_2 = 2, L_1 = -1 $ and $L_2 = 3$. Pop-up bifurcation occurs at $K_1^{\pu} = 3.9175$ and $(r_1^{\pu}, r_2^{\pu}) \approx (0.5699,0.8325)$. In addition, $K_1^{\sym} = 6$. }\label{fig:bpopl2}
\end{center}
\end{figure}

\begin{figure*}
        \centering
        \begin{subfigure}[b]{0.45\textwidth}   
            \centering 
            \includegraphics[width=\linewidth]{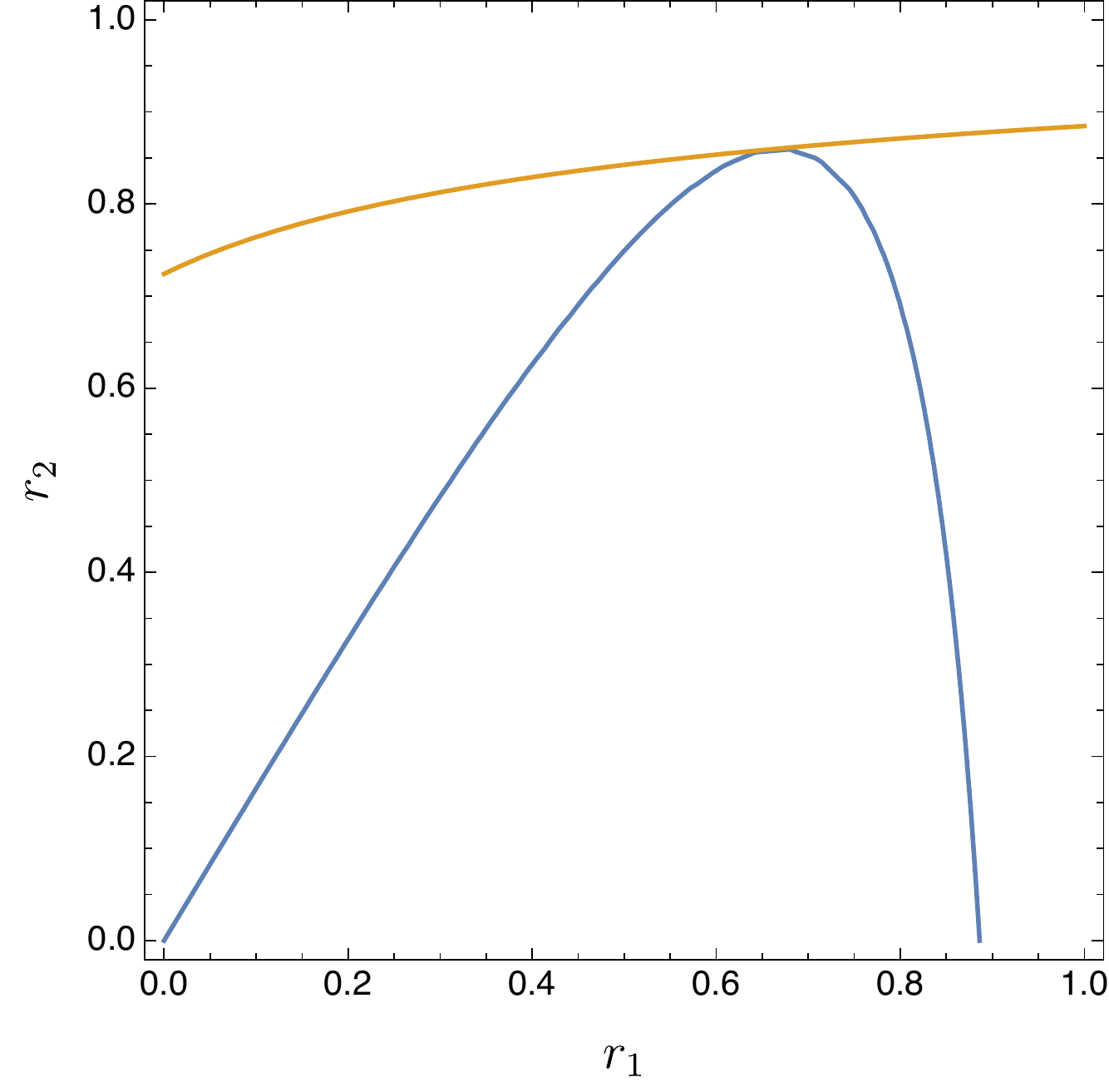}
            \caption[]%
            {{\small $K_1 = 5.316,~ K_2 = 3,~ L_1 = -2,~ L_2 = 2$.}}    
            \label{fig:Indir3}
        \end{subfigure}
        \hfill
        \begin{subfigure}[b]{0.45\textwidth}   
            \centering 
            \includegraphics[width=\linewidth]{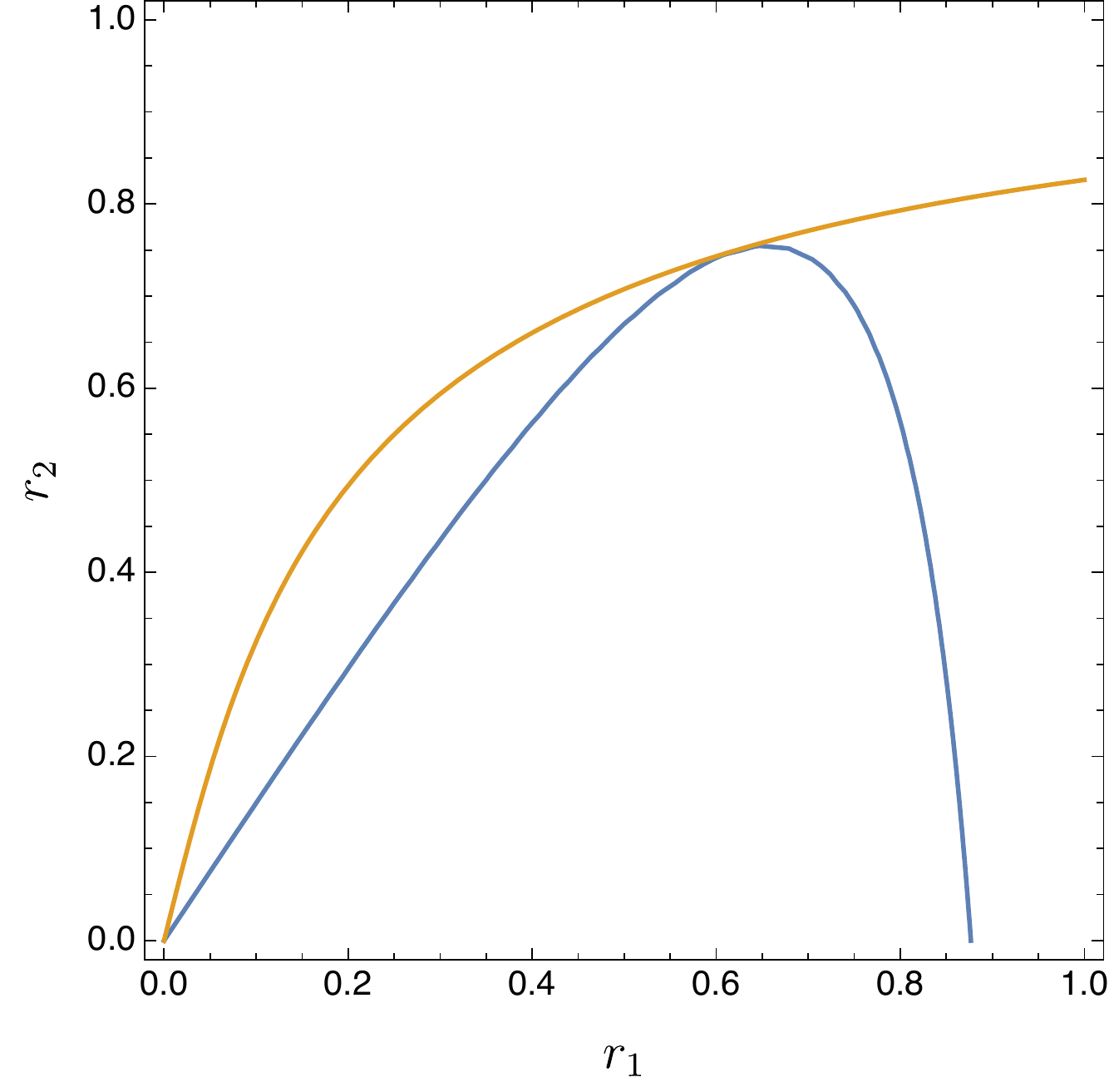}
            \caption[]%
            {{\small $K_1 = 4.999,~ K_2 = 1.5,~ L_1 = -2,~ L_2 = 2$. }}    
            \label{fig:Indir4}
        \end{subfigure}
         \vskip\baselineskip
        \begin{subfigure}[b]{0.45\textwidth}
            \centering
            \includegraphics[width=\linewidth]{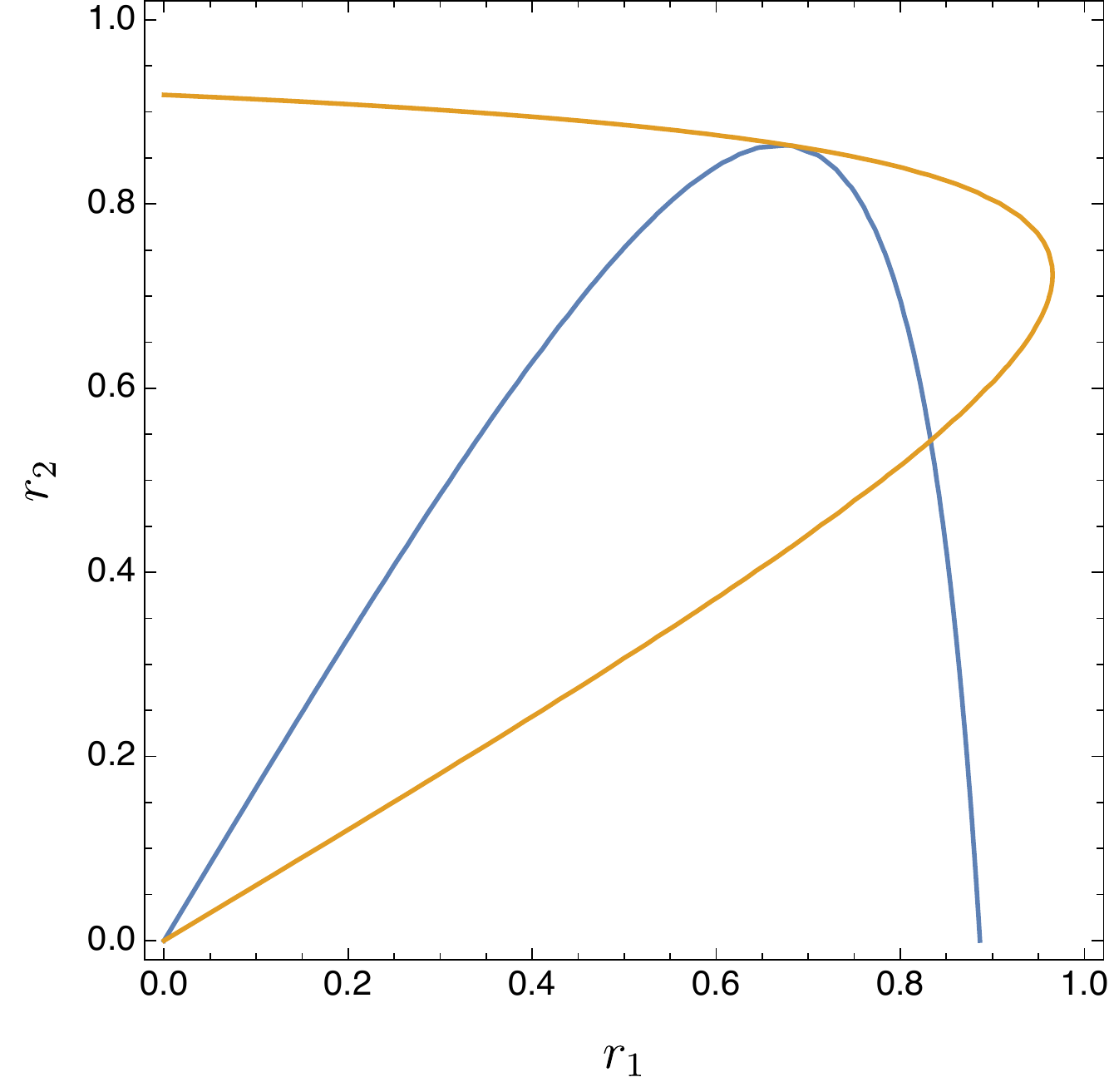}
            \caption[]%
            {{\small $K_1 = 5.329,~ K_2 = 7,~ L_1 = -2,~ L_2 = -3$.}}    
            \label{fig:Indir1}
        \end{subfigure}
        \hfill
        \begin{subfigure}[b]{0.45\textwidth}  
            \centering 
            \includegraphics[width=\linewidth]{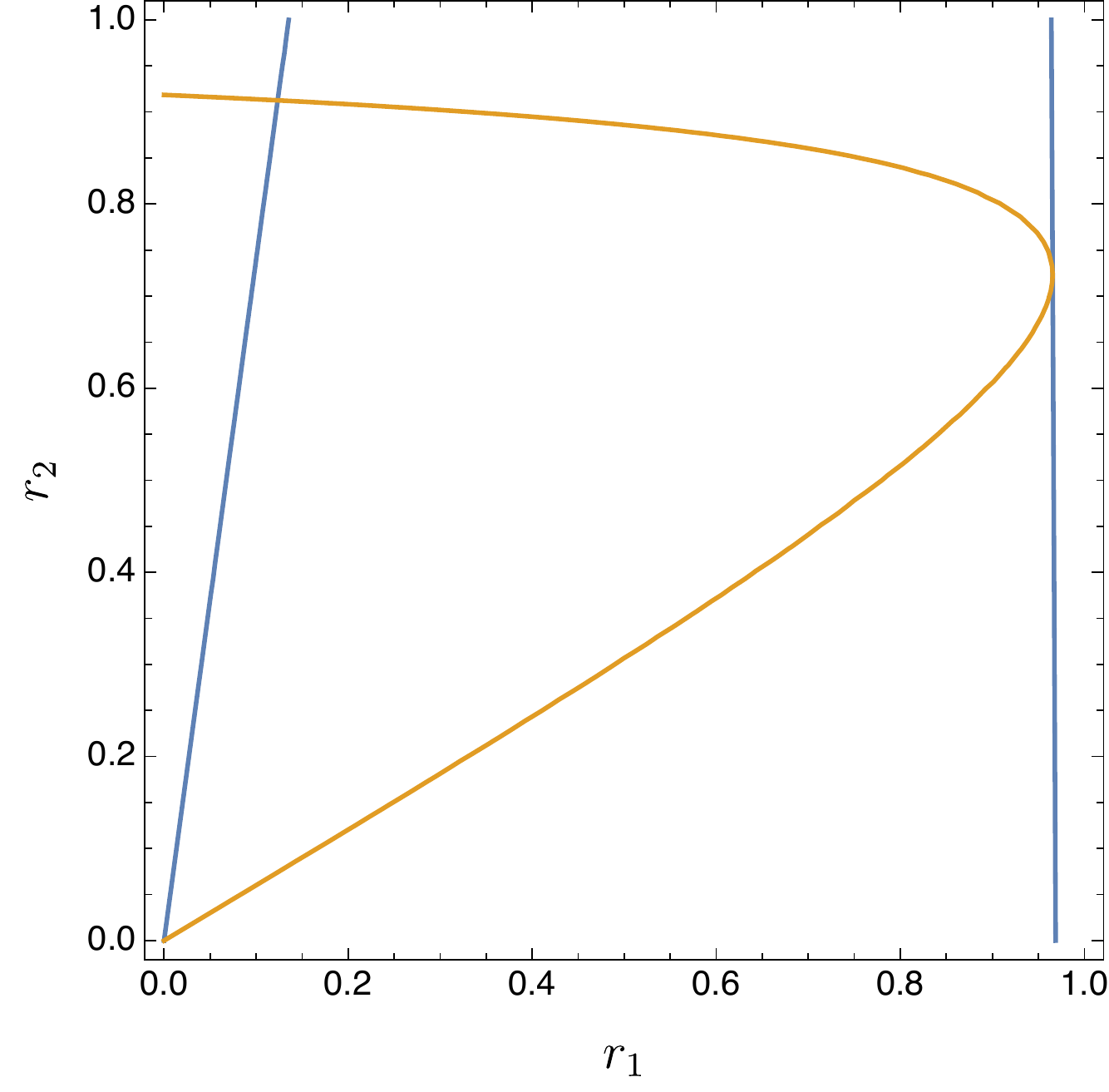}
            \caption[]%
            {{\small $K_1 = 16.804,~ K_2 = 7,~ L_1 = -2,~ L_2 = -3$.}}    
            \label{fig:Indir2}
        \end{subfigure}        
        \caption{ Examples of possible geometric configurations with $\p \Gamma^{K_1, L_1}_1/\p r_1(r_1, r_2) = \p \Gamma^{K_2, L_2}_2/\p r_1 (r_1,r_2)$, with $(r_1, r_2) \neq (0,0)$. Note that we require that at least one of the level curves has a turning point, which means that either $L_1 < 0$ or $L_2 < 0$ (Property 1 of \cite[Theorem III.11]{Achterhof2020}). If $L_1$ and $L_2$ have a opposite sign and if we fix three of the four interaction strengths and varying the remaining interaction strength, then there exists at most one point where the derivatives equal (see Figure \ref{fig:Indir3}, Figure \ref{fig:Indir3} and Lemma\ref{lem:uniqsol}). Furthermore, in the case where $L_1 < 0$ and $L_2 < 0$ there are two possibilities where the derivatives equal. E.g. in Figure \ref{fig:Indir1} and Figure \ref{fig:Indir2} we see that for fixed $K_2, L_1$ and $L_2$ there exists two possible values for $K_1$ such that the derivatives equal, namely $K_1^* = 5.329$ and $K_1^* = 16.804$. This is true because if $L_1 < 0$ and $L_2 < 0$, then both level curves have a turning point (see Lemma \ref{lem:nobf}).}\label{fig:eqdir}
\end{figure*}

\subsubsection{The inter-community interaction strengths have opposite sign}
We assume that $L_1 < 0$ and $L_2 > 0$ or $L_1 > 0$ and $L_2 < 0$ and take $K_{1}$ and $K_{2}$ to be such that the self-consistency surfaces are not trivial. In terms of the fundamental curves this corresponds to the situation where the solutions to the self-consistency equations are given by the intersection points of a parabola and a line connected with zero. See Figure \ref{fig:asympillustration} for a numerical example in this case. 

\begin{figure*}
        \centering
        \begin{subfigure}[t]{0.45\textwidth}
            \centering
            \includegraphics[width=\textwidth]{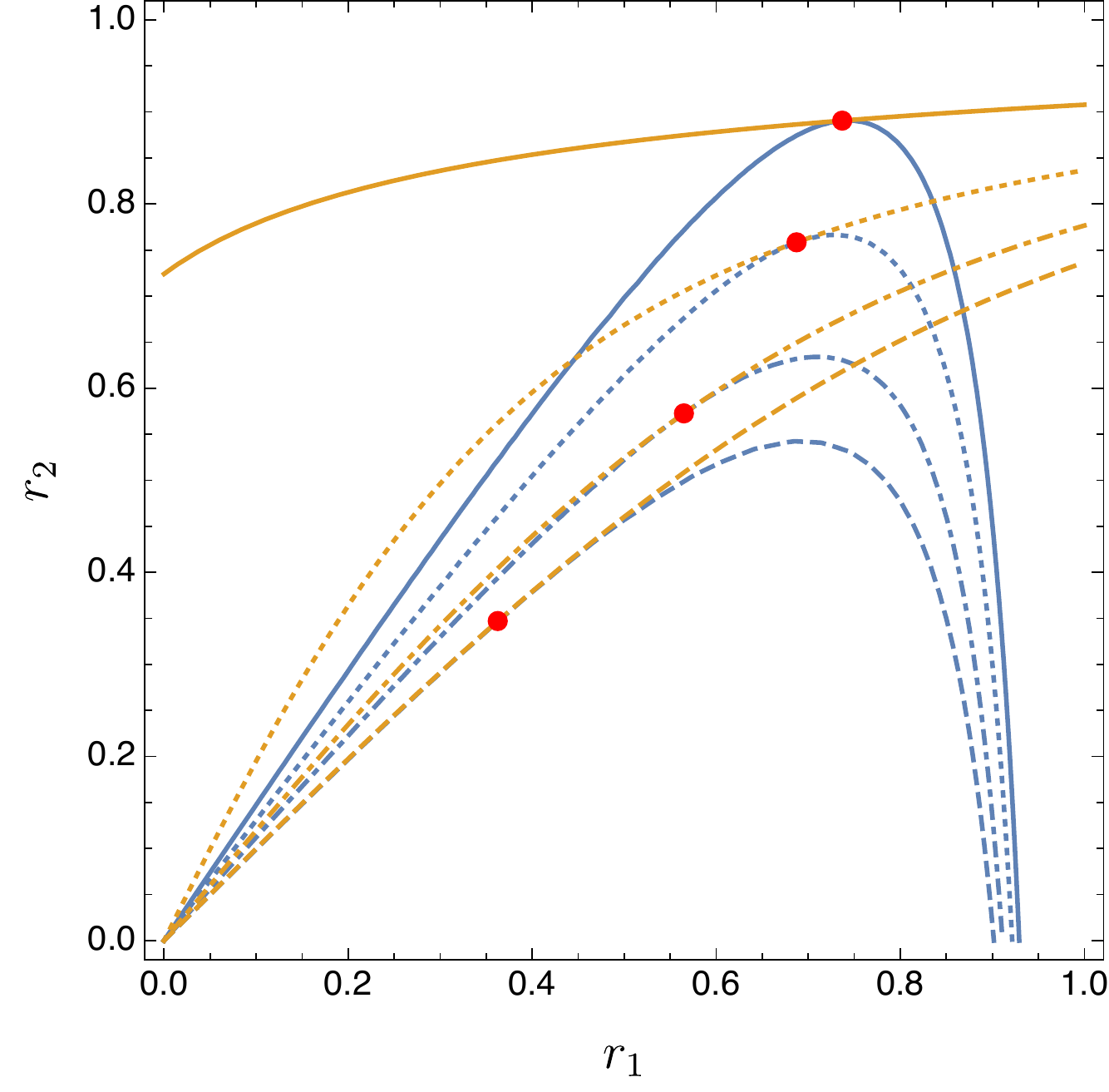}
            \caption{Plot of the level curves $\Gamma_1$ and $\Gamma_2$ with $L_1 = -4$, $L_2 = 3$ and for different pairs of interaction strengths $(K_1, K_2)$ such that $\p \Gamma_1^{K_1, L_1}/\p r_1 = \p \Gamma_2^{K_2,L_2}/\p r_1$ for some $(r_1, r_2)$. The touching point is denoted by a red dot.}\label{fig:inbif1}
        \end{subfigure}
        \hfill
        \begin{subfigure}[t]{0.45\textwidth}  
            \centering 
            \includegraphics[width=\textwidth]{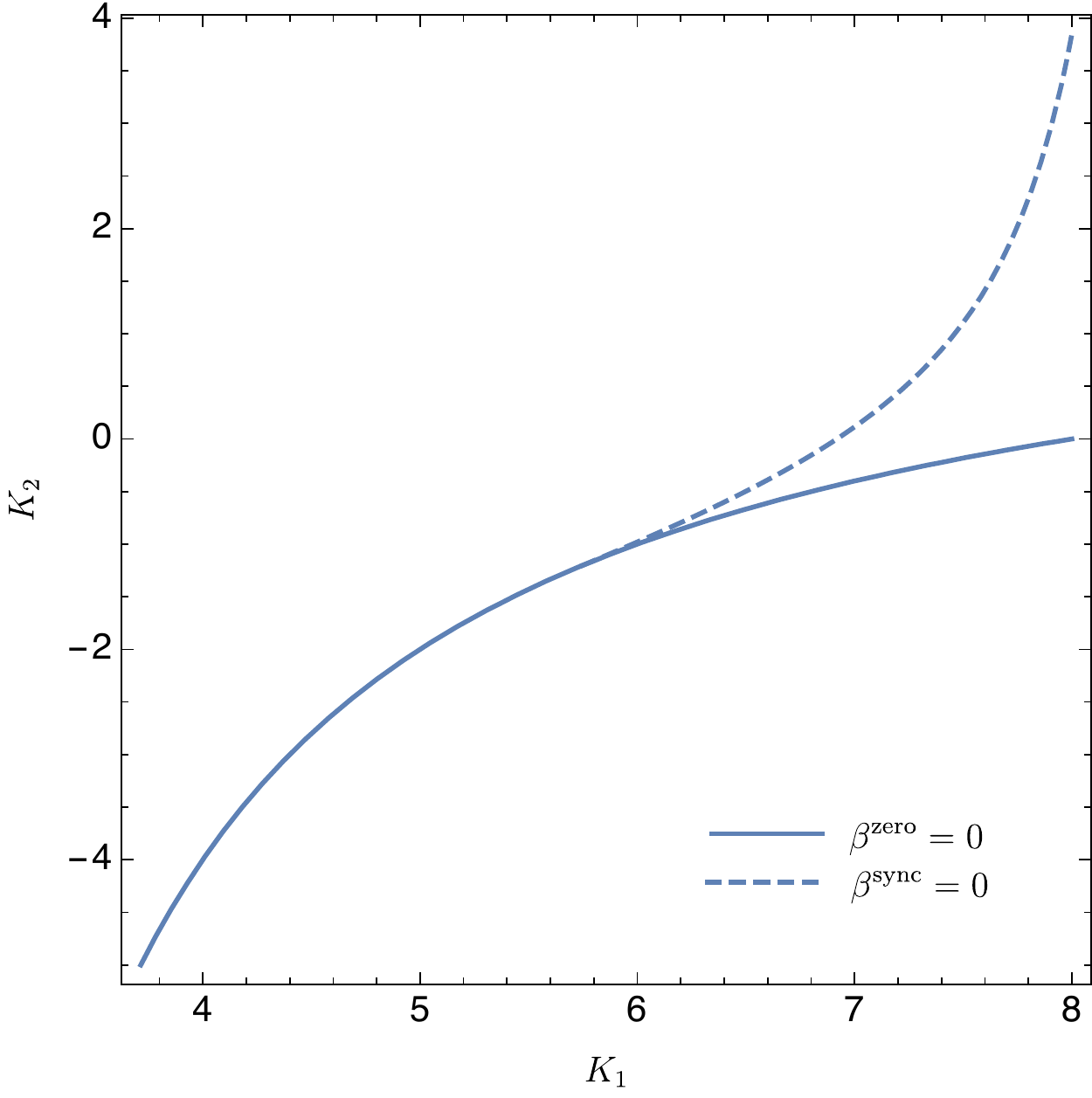}
            \caption{Plot of the solution boundaries $\beta^{\zero} = 0$ and $\beta^{\sync} = 0$ with $L_1 = -4$, $L_2 = 3$ and where $K_1$ and $K_2$ are varied. The solution boundary $\beta^{\sync} = 0$ bifurcates from $\beta^{\zero} = 0$ at $K_1 = 5.747$, $K_2 = -1$.}\label{fig:inbif2}
        \end{subfigure}
        \caption{A numerical example in the case where $L_1$ and $L_2$ have opposite sign (due to Theorem \ref{thm:opsign}). In this case $L_1 = -4$ and $L_2 = 3$.  On the right-hand side a plot of the solution boundary $\beta^{\zero} = 0$ (see \cite[Theorem IV.1]{Achterhof2020}) and the solution boundary $\beta^{\sync} = 0$ (see Theorem \ref{thm:biffromsync}). On the left-hand different pairs $(K_1, K_2)$ are taken that lie on the solution boundary $\beta^{\sync} = 0$. In this example we have from the top to the bottom: $(K_1, K_2) = (7.901, 3), (7.234, 0.5), (6.497, -0.5), (5.977, -1)$. The intersection points are: $(r_1, r_2) = (0.565, 0.573), (0.687,0.758), (0.737, 0.891), (0.363, 0.347)$. We see that $(r_1, r_2)$ decreases to $(0,0)$ as $(K_1, K_2)$ decreases. This corresponds with the figure on the right-hand side because the solution boundary $\beta^{\sync} = 0$ bifurcates from $\beta^{\zero} = 0$. At the solution boundary $\beta^{\zero} = 0$ there is one possible solution, namely the unsynchronized solution $(r_1, r_2) = (0,0)$.} 
\label{fig:asympillustration}
\end{figure*}

\begin{definition}[Boundary set]
We define the \emph{boundary set} for $K_1$ as
\begin{equation}
K_1^*(K_2,L_1,L_2) := \{ K_1 \in \R: \beta^{\sync}(K_1, K_2, L_1, L_2) = 0 \},
\end{equation}
and define $K_2^*, L_1^*$ and $L_2^*$ analogously. 
\end{definition}

We will show that $K_1^*, K_2^*, L_1^*$ and $L_2^*$ contain precisely one element when $L_1$ and $L_2$ have opposite sign.

\begin{lemma}[Uniqueness of the solution boundary]
\label{lem:uniqsol} 
Suppose that $L_1$ and $L_2$ have opposite sign. Then $K_1^*, K_2^*, L_1^*$ and $L_2^*$ contain at most one element.
\end{lemma}

\begin{proof}
Fix $K_2, L_1$ and $L_2$ such that $L_1$ and $L_2$ have opposite sign and $\Gamma_2^{K_2, L_2}$ is non-trivial. In this case, $\Gamma_1$ or $\Gamma_2$ is a parabola (but not both). This means that there is at most one point $K_1^*$ for which $\p \Gamma_1^{K_1^*, L_1}/\p r_1 = \p \Gamma_2^{K_2, L_2}/\p r_1$ (due to the geometry of the relevant fundamental curves). Note that the latter condition is equivalent to $\beta^{\sync}(K_1^*, K_2, L_1, L_2) = 0$. We can repeat this geometric reasoning for $K_2^*, L_1^*$ and $L_2^*$.
\end{proof}

Due to the uniqueness shown in Lemma \ref{lem:uniqsol} we can define the \emph{boundary functions} $K_2 \mapsto K_1^*(K_2; L_1, L_2)$, $L_1 \mapsto K_1^*(L_1; K_2, L_2)$ and $L_2 \mapsto K_1^*(L_2; K_2, L_1)$. In a similar way we define the boundary functions $K_2^*(\cdot), L_1^*(\cdot)$ and $L_2^*(\cdot)$. We analyze these boundary functions by investigating their asymptotes.

\begin{definition}[Asymptotes]\label{def:asymp} Assume that $\Gamma^{K_1, L_1}_1$ and $\Gamma^{K_2, L_2}_2$ are non-trivial. We define the following asymptotes.

\begin{enumerate}
\item If $L_1 < 0$, then the asymptotes $K_1^a(L_1)$ and $L_1^a(K_1)$ are the unique solutions of the system of equations
\begin{equation}
V(K_1 r + L_1) - r = 0, \quad  V\left(\frac{K_1 r}{K_1 (1 - r^2) - 1} \right) - r = 0, \label{eq:defas1}
\end{equation}
with respect to $K_{1}$ and $L_{1}$ respectively (for some $r \in (0,1)$). 

\item If $L_2 < 0$, then asymptotes $K_2^a(L_2)$ and $L_{2}^{a}(K_2)$ are the unique solutions of the system of equations 
\begin{equation}
V(K_2 s + L_2) - s = 0, \quad  V\left(\frac{K_2 s}{K_2 (1 - s^2) - 1} \right) - s = 0, \label{eq:defas2}
\end{equation}
with respect to $K_{2}$ and $L_{2}$ respectively (for some $r \in (0,1)$). 

\item If $L_1 < 0$, then the asymptotes $K_1^b(K_2, L_1)$ and $L_1^b(K_1, K_2)$  are the unique solutions of 
\begin{align}
r_2 &= \frac{- K_1^2 r_1^3+ K_1^2 r_1-2 K_1 r_1}{L_1 \left(K_1 r_1^2-K_1+1\right)},\nonumber\\
r_1 &= V\left( \frac{ K_1 r_1}{K_1 (1 - r_1^2) - 1} \right),\label{eq:defas3}\\
r_2 &= V(K_2 r_2), \nonumber
\end{align}
with respect to $K_{1}$ and $L_{1}$ respectively (for some $r_1, r_2 \in (0,1)$).

\item If $L_2 < 0$, then the asymptotes $K_2^b(K_1, L_2)$ and $L_2^b(K_1, K_2)$ are the unique solutions of 
\begin{align}
r_1 &= \frac{-K_2^2 r_2^3+K_2^2 r_2-2 K_2 r_2}{L_2 \left(K_2 r_2^2-K_2+1\right)},\nonumber\\
r_2 &= V\left( \frac{ K_2 r_2}{K_2(1 - r_2^2) - 1} \right),\label{eq:defas4}\\
r_1 &= V(K_1 r_1),\nonumber 
\end{align}
with respect to $K_{2}$ and $L_{2}$ respectively (for some $r_1, r_2 \in (0,1)$).
\end{enumerate}
\end{definition}

\begin{remark}[Geometric interpretation of Definition \ref{def:asymp}]
\label{rem:interp}
In order to clarify \eqref{eq:defas1}-\eqref{eq:defas4} we give the following geometric interpretation. In the proof of Property 1 of \cite[Theorem III.11]{Achterhof2020} we showed that for $L_1 < 0$ the fundamental curve $\Gamma_1^{K_1,L_1}$ is a parabola and the turning point $(r_1, r_2)$ uniquely solves the equations
\begin{equation}
r_1 = V\left( \frac{ K_1 r_1}{K_1(1 - r_1^2) - 1} \right) \label{eq:topcoordinates1}
\end{equation}
and 
\begin{equation}
\label{eq:topcoordinates2}
r_2 = \frac{-K_1^2 r_1^3+K_1^2 r_1-2 K_1 r_1}{L_1 \left(K_1 r_1^2-K_1+1\right)}.
\end{equation}
Now, the geometric configuration corresponding to \eqref{eq:defas1} is that the top of the level curve $\Gamma_1$ intersects the line $[0,1] \times \{1\}$, i.e. we take $r_2 = 1$ in \eqref{eq:topcoordinates1} and \eqref{eq:topcoordinates2}. A numerical example of this geometric configuration is given in Figure \ref{fig:Lnegstart}. The geometric configuration corresponding to \eqref{eq:defas2} is the same with the roles of the fundamental curves reversed. Next, a numerical example of the geometric configuration behind \eqref{eq:defas3} is given in Figure \ref{fig:Lnegb3}. The top of the parabola $\Gamma_1$ touches the vertical line drawn from the intersection point of $\Gamma_2$ with the line $\{0\} \times [0,1]$.
\end{remark}

\begin{remark}[Existence of the asymptotes] Note that by the geometric interpretation of the asymptotes (Remark \ref{rem:interp}) the asymptotes in Definition \ref{def:asymp} exist and are uniquely determined. In addition, we have 
\begin{align}
K_1^b(K_2, L_1) < K_1^a(L_1), \quad	& L_1^b(K_1, K_2) < L_1^a(K_1),\\
K_2^b(K_1,L_2) < K_2^a(L_2),  \quad	& L_2^b(K_1,K_2) < L_2^a(K_2).
\end{align}

\end{remark}

\begin{theorem}[Asymptotes when $L_1$ and $L_2$ have opposite signs]\label{thm:opsign}
Suppose that $L_1$ and $L_2$ have opposite signs. If $L_1 < 0$ and $L_2 > 0$, then the asymptotes of the boundary functions are given in Table \ref{fig:astab1}. If $L_1 > 0$ and $L_2 < 0$, then the asymptotes of the boundary functions are given in Table \ref{fig:astab2}. Inside both tables the limit values of the boundary functions are given as one of the four interaction strengths tends to one of the values displayed in the left column. E.g. the $``0"$ in row $2$ column $1$ of Table \ref{fig:astab1} corresponds with the limit $\lim_{K_1 \to 2} L_2^*(K_1; K_2, L_1, L_2) = 0$. 
        
\begin{table}[!ht]
\parbox{.45\linewidth}{ \centering
\begin{tabular}{l l l l l} \toprule
Boundary func				& $K_1^*$  & $K_2^*$ & $L_1^*$ & $L_2^*$ \\ \midrule
$K_1 \to K_1^a$		& & $\infty$ & & $\infty$ \\
$K_1 \to K_1^b$		& & & & $0$\\ \hline

$K_2 \to \infty$ 	& $K_1^a$ & & $L_1^a$ & \\ \hline

$L_1 \to L_1^a$		& & $\infty$ & & $\infty$ \\
$L_1 \to L_1^b$		& & & & $0$ \\ \hline

$L_2 \to 0$ 		& $K_1^b$ & & $L_1^b$ & \\
$L_2 \to \infty$	& $K_1^a$ & & $L_1^a$ & \\ \bottomrule
\end{tabular}
\caption{The asymptotes when $L_1 < 0$ and $L_2 > 0$.\newline}\label{fig:astab1}
}
\hfill
\parbox{.45\linewidth}{ \centering
\begin{tabular}{l l l l l} \toprule
Boundary func				& $K_1^*$  & $K_2^*$ & $L_1^*$ & $L_2^*$ \\ \midrule
$K_1 \to \infty$ 	& &$K_2^a$  & & $L_2^a$  \\ \hline

$K_2 \to K_2^a$		& $\infty$ &  & $\infty$  &\\
$K_2 \to K_2^b$		& & & $0$ &\\ \hline

$L_1 \to 0$ 		& & $K_2^b$  & & $L_2^b$ \\
$L_1 \to \infty$	& & $K_2^a$ & & $L_2^a$ \\ \hline

$L_2 \to L_2^a$		& $\infty$ & & $\infty$ & \\
$L_2 \to L_2^b$		& & & $0$ & \\ \bottomrule
\end{tabular}
\caption{The asymptotes when $L_1 > 0$ and $L_2 < 0$.}\label{fig:astab2}
}
\end{table}
\end{theorem}
The proof of Theorem \ref{thm:opsign} is given in Appendix \ref{app:opsign}.

For a visualization of some asymptotics see Figure \ref{fig:asymptoticillustration}. 

\subsubsection{The inter-community interaction strengths are both negative}

We consider the case where $L_1 < 0$ and $L_2 < 0$. By \cite[Theorem III.3]{Achterhof2020} this implies that $K_1 \geq 2$ and $K_2 \geq 2.$ In terms of the geometry of the level curves this corresponds to the intersection of two parabolas. This case is harder to analyze than the case in the previous subsection where $L_1$ and $L_2$ have opposite sign due to $K_1^*$ possibly containing multiple elements. Illustrations of the arguments of this section are shown in Figure \ref{fig:Lnegb} and Figure \ref{fig:Lnegstart}. 

\begin{lemma}\label{lem:nobf}
If $L_1 < 0$ and $L_2 < 0$, then $K_1^*$ contains at most two elements. The same is true for $K_2^*, L_1^*$ and $L_2^*$.
\end{lemma}

\begin{proof}
Fix $K_2 > 2, L_1$ and $L_2$ such that $L_1, L_2 < 0$. In this case $\Gamma_1$ and $\Gamma_2$ are parabolas. In this at most two points $K_1^*$ for which $\p \Gamma_1^{K_1^*, L_1}/\p r_1 = \p \Gamma_2^{K_2, L_2}/\p r_1$ exist.
\end{proof}

If $\# K_1^* = 2$, then we denote by $K_{1,+}^*$ and $K_{1,-}^*$ the two elements of $K_1^*$ and we set $K_{1,+}^* > K_{1,-}^*$. By Lemma \ref{lem:nobf} we can construct (for fixed $L_1$ and $L_2$) a boundary function as follows:

\begin{equation}
K_2 \mapsto 
\begin{cases}
(K_{1,+}^*(K_2), K_{1,-}^*(K_2)) 	& \text{if} \quad \# K_1^*(K_2) = 2,\\
K_1^*(K_2) 							& \text{if} \quad \# K_1^*(K_2) = 1,\\
\emptyset							& \text{if} \quad \# K_1^*(K_2) = 0.
\end{cases}
\end{equation}
Since $K_{1,+}^*$ and $K_{1,-}^*$ are both part of a boundary function, they must coincide at some point. This is true because a solution boundary isolates two regions with a different number of solutions. We call the point where $K_{1,+}^*$ and $K_{1,-}^*$ coincide the \emph{starting point} of the boundary function.

\begin{proposition}[Starting point]
The following are equivalent: 

\begin{enumerate}
\item $K_{1,+}^* = K_{1,-}^*$,

\item The boundary functions $K_{1,+}^*, K_{1,-}^*$ both solve (\ref{eq:Lsc1}) with respect to $K_{1,+}^* = K_1^s$ and $K_{1,-}^* = K_1^s$:
\begin{equation}
C_{1,1} = \frac{1}{K^s_1 - L^s_1}, \quad C_{2,1} = \frac{1}{K^s_2 - L^s_2}, \label{eq:Lsc1}
\end{equation}
for some $(r_1, r_2) \in \Gamma_1^{K_1^s, K_2^s} \cap \Gamma_2^{K_2^s, L_2^s}.$
\end{enumerate}
\end{proposition}

\begin{proof}
At the starting point we have (see Figure \ref{fig:Lnegstart})
\begin{equation}
\frac{\p \Gamma_1}{\p r_1} = \frac{\p \Gamma_2}{\p r_1},\quad \frac{\p \Gamma_1}{\p r_1} = \frac{\p \Gamma_1}{\p r_2}, \quad \frac{\p \Gamma_2}{\p r_1} = \frac{\p \Gamma_2}{\p r_2}, 
\end{equation}
which implies that 
\begin{equation}
\pm L_1 C_{1,1} = 1 - K_1 C_{1,1}, \quad \pm L_2 C_{2,1} = 1 - K_2 C_{2,1}. \label{eq:Lsc2}
\end{equation}
If we take the plus sign, then 
\begin{equation}
\frac{\p \Gamma_1}{\p r_1} = \frac{\p \Gamma_2}{\p r_1} = 1,
\end{equation}
which is not possible at the starting point (see Figure \ref{fig:Lnegstart}). Hence we take the minus sign, and by rewriting \eqref{eq:Lsc2} we get \ref{eq:Lsc1}. Furthermore, by the geometry of the level curves, the uniqueness follows.
\end{proof}

\begin{theorem}[Asymptotes when $L_1 < 0$ and $L_2 < 0$]\label{thm:negsign12}
Suppose that $L_1 < 0$ and $L_2 < 0$. 

\begin{enumerate}
\item The asymptotes of $K_{1,+}^*, K_{2,-}^*, L_{1,-}^*$ and $L_{2,+}^*$ are given in Theorem \ref{thm:opsign}(1).
\item The asymptotes of $K_{1,-}^*, K_{2,+}^*, L_{1,+}^*$ and $L_{2,-}^*$ are given in Theorem \ref{thm:opsign}(2).
\end{enumerate}
\end{theorem}

\begin{proof}
The proof is analogous to the proof of Theorem \ref{thm:opsign}.
\end{proof}

\begin{figure*}
        \centering
        \begin{subfigure}[t]{0.45\textwidth}  
            \centering 
            \includegraphics[width=\textwidth]{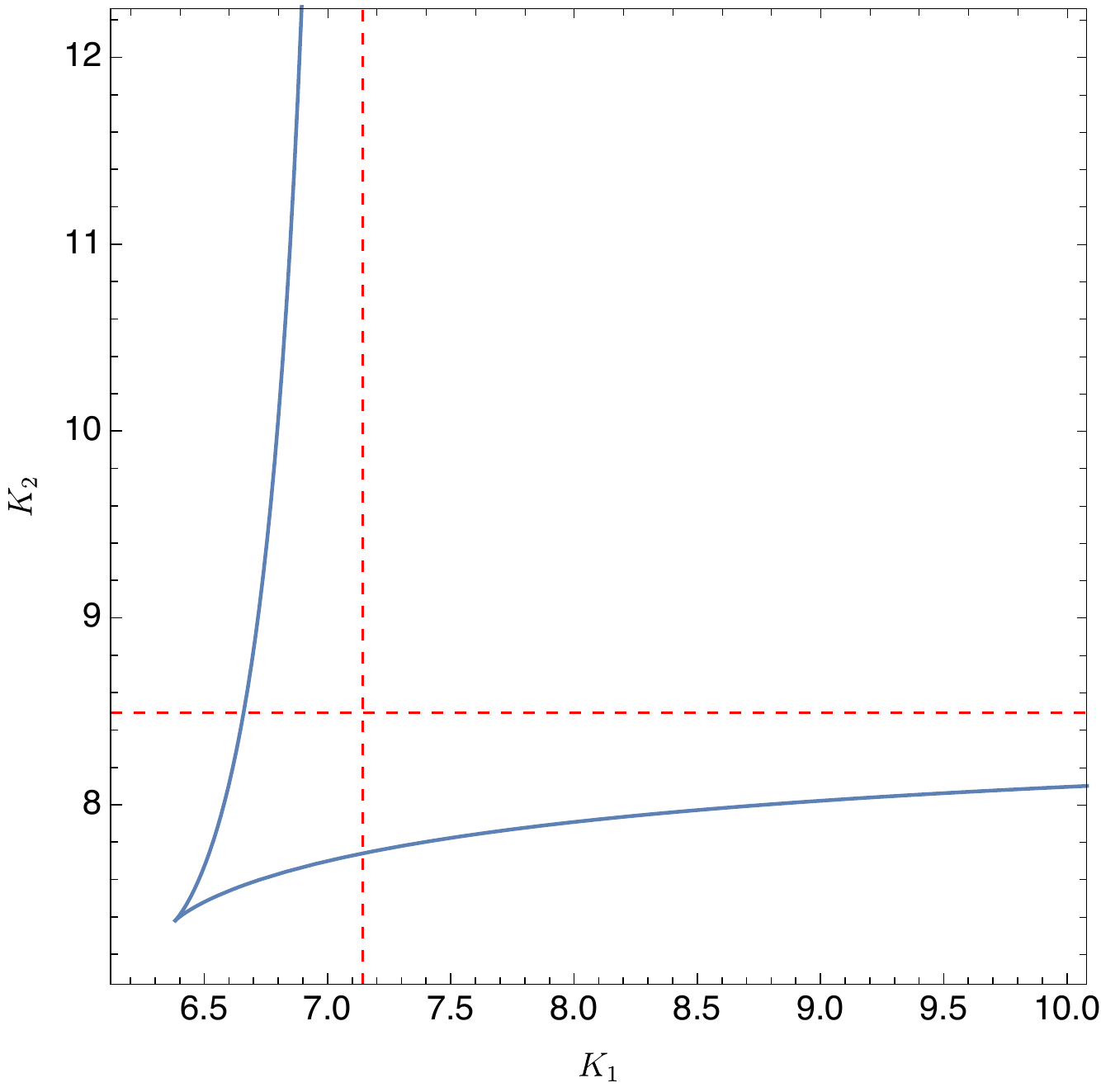}
            \caption{ Plot of the boundary function $K_2^*(K_1)$ with $L_1 = -3$, $L_2 = -4$. The starting point is at $K_1^s = 6.382$, $K_2^s = 7.381$. Furthermore, there is a vertical asymptote for $K_2^*(K_1)$ at $K^a_1 = 7.143$ and a horizontal asymptote at $K_2^a = 8.492$. }    
            \label{fig:Lnegb}
            \end{subfigure}
            \hfill
        \begin{subfigure}[t]{0.45\textwidth}
            \centering
            \includegraphics[width=\textwidth]{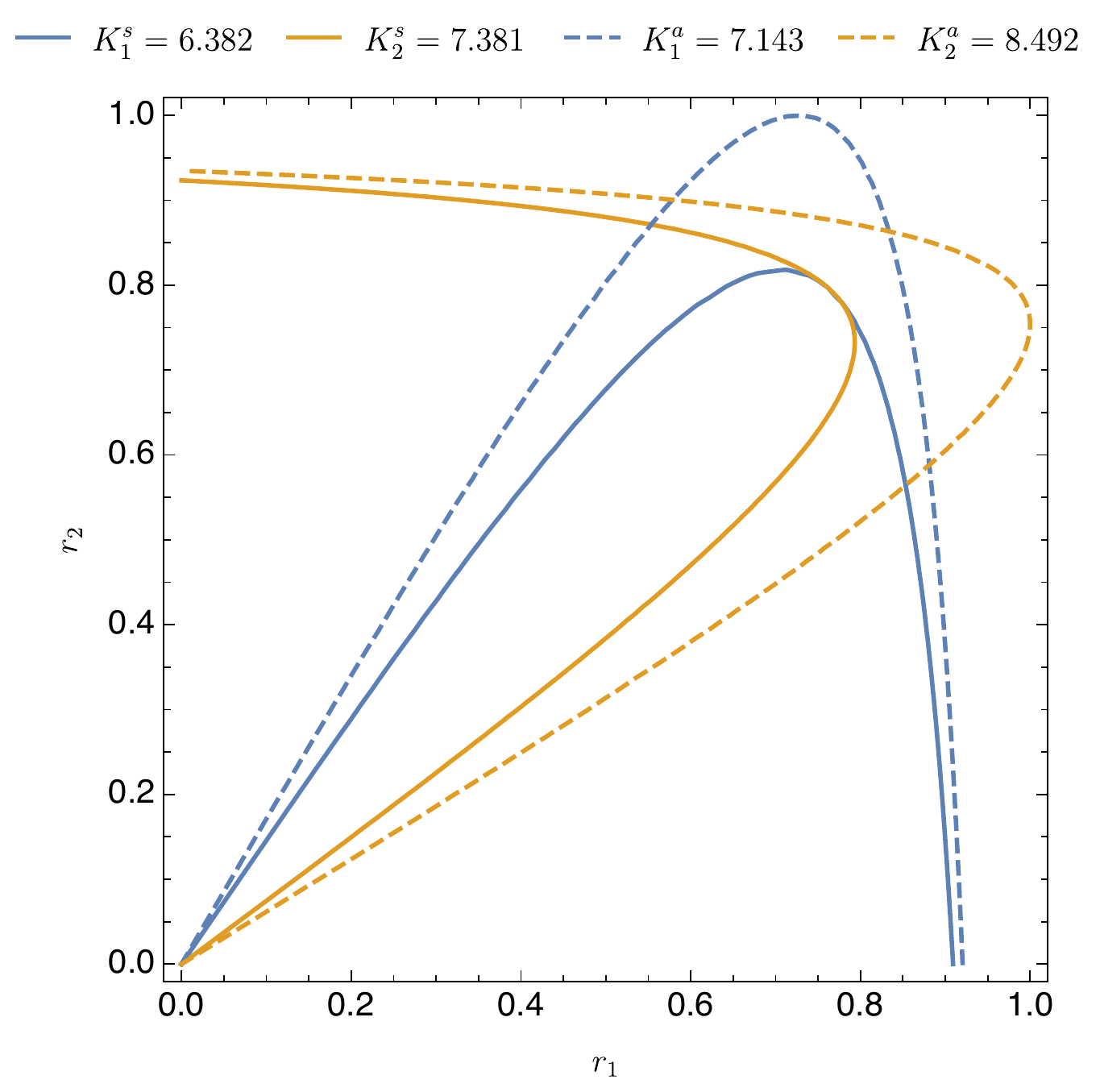}
            \caption
            { Plot of the level curves $\Gamma_1$ and $\Gamma_2$ with $L_1 = -3, L_2 = -4$ and $K_1, K_2$ varied. The starting point is at the intersection of the two solid level curves.}    
            \label{fig:Lnegstart}
        \end{subfigure}

        \begin{subfigure}[t]{0.45\textwidth}
            \centering
            \includegraphics[width=\textwidth]{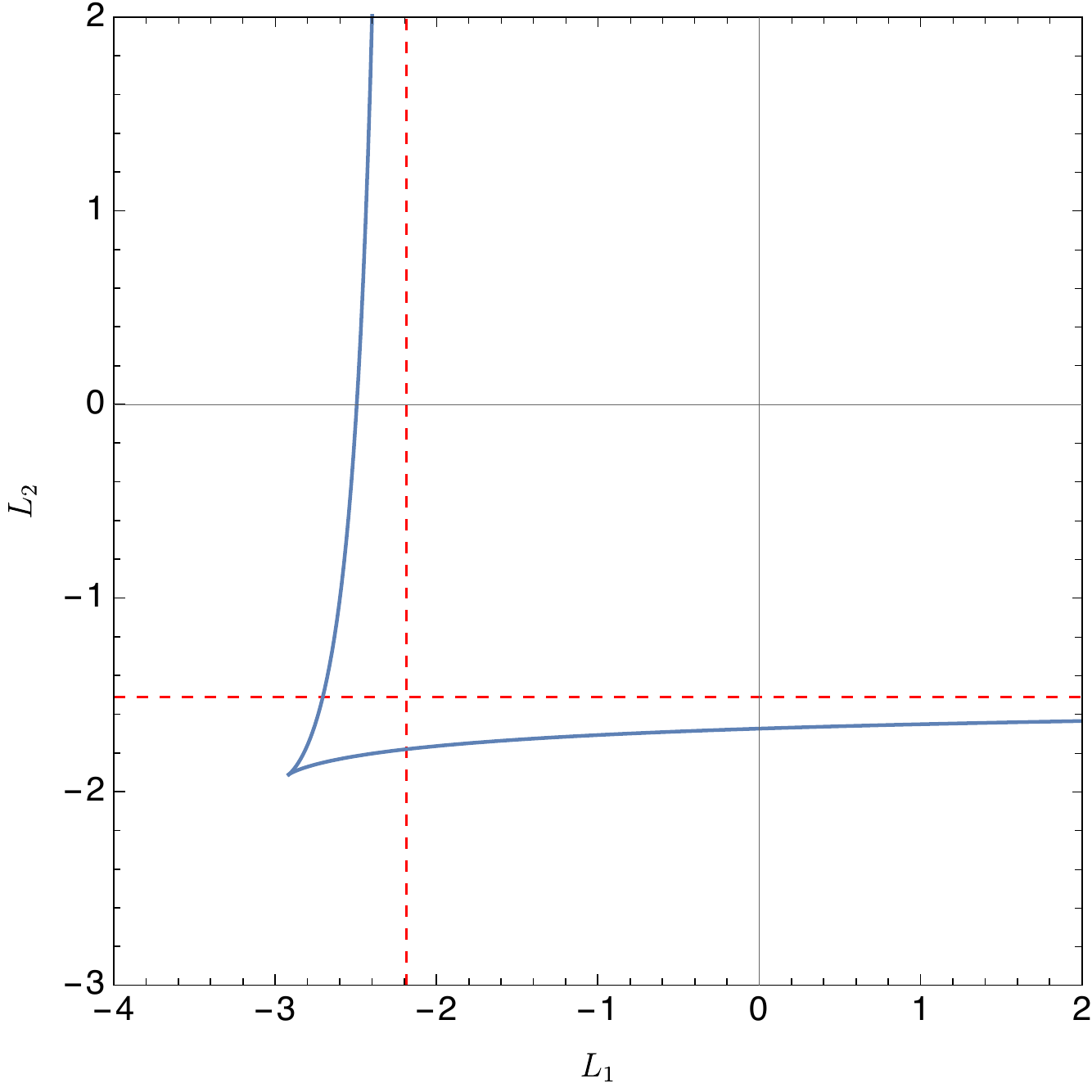}
            \caption
            { Plot of the boundary function $L_2^*(L_1)$ with $K_1 = 6$, $K_2 = 5$. The starting point is at $L_1^s = 2.916$, $L_2^s = -1.911$. There is a vertical asymptote for $L_2^*(L_1)$ at $L^a_1 = -2.187$ and a horizontal asymptote at $L_2^a = -1.511$. In addition, the asymptote $L_1^b = -2.494$ is at the intersection of $L_2^*$ with $L_2 = 0$ and the asymptote $L_2^b = -1.675$ is at the intersection of $L_2^*$ with $L_1 = 0$.   }    
            \label{fig:Lnegb2}
        \end{subfigure}
        \hfill
        \begin{subfigure}[t]{0.45\textwidth}  
            \centering 
            \includegraphics[width=\textwidth]{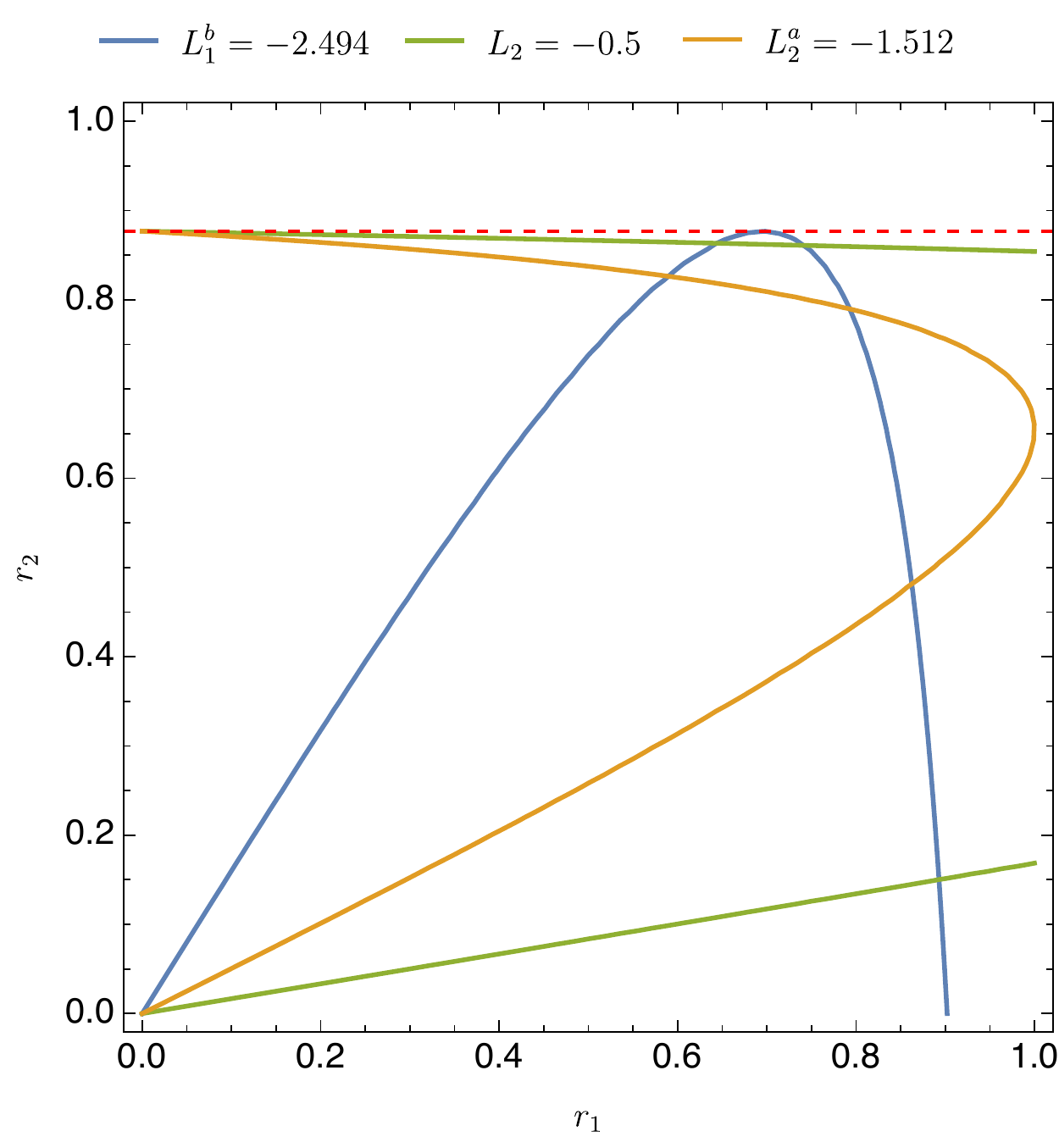}
            \caption{ Plot of the level curves $\Gamma_1$ and $\Gamma_2$ with $K_1 = 6, K_2 = 5$ and $L_1, L_2$ varied. The level curve $\Gamma_1$ at $L_1^b$ is displayed. The red dashed line is tangent to the top of $\Gamma_1$ and the intersection point of $\Gamma_2$ with $r_1 = 0$. As $L_2 \to 0$ the top of curve $\Gamma_2$ converges to this tangent line.}    
            \label{fig:Lnegb3}
        \end{subfigure}
                
\caption{ On the left-hand side, $K_2^*(K_1)$ and $L_2^*(L_1)$ are plotted and on the right-hand side, the geometric configurations at the asymptotes are given. The definitions of the asymptotes $K_1^a, K_2^a, L_2^a$ and $L_1^b, L_2^b$ are given in Definition \ref{def:asymp}. Furthermore, the starting point $K_1^s, K_2^s, L_1^s$ and $L_2^s$ is defined in \eqref{eq:Lsc1}.}
\label{fig:asymptoticillustration}
\end{figure*}

\section{Classification of the number of solutions}
\label{sec:classification}
In this section we will use the insight from the  geometric description of the self-consistency equations (allowing us to split the parameter space into various regions) and the identification of solution boundaries, to give a full classification of the number of possible solutions to the self-consistency equations in the entire parameter space. Using only the fundamental curves and geometric arguments, we can tabulate the maximum number of solutions per region as in Table \ref{fig:overview}.

In the following subsections we will refine this table by making use of the expressions for the solution boundaries $\beta^{\sync}$, $\beta^{\psync}$ and $\beta^{\zero}$, the conditions for their existence and the asymptotes identified in Theorem \ref{thm:opsign} and Theorem \ref{thm:negsign12}. 

We start the classification by giving all the possible bifurcation types in each non-trivial region (see Table \ref{fig:bifoverview}).

\begin{remark}[Overview of bifurcation types]
Based on the results of the previous sections we determine in each of the non-trivial regions what kind of bifurcation occurs. The occurrence of a bifurcation from zero is characterized in Lemma \ref{lem:nobif} and Theorem \ref{thm:bifzero1}. The occurrence of a bifurcation from a limit point is characterized in Theorem \ref{thm:biflim}. Furthermore, the occurrence of a bifurcation from a synchronized solution is characterized in Lemma \ref{lem:solempty} and Lemma \ref{lem:solnonempty}.

\begin{table}[!ht]
\centering
\vspace{0.1cm}
\begin{tabular}{l c c c c c} \toprule
\textbf{Region}	& $2$  & $3,4,5$ & $6,7$ & $8,9$ & $10$ \\ \midrule
$\beta^{\zero}$		& Yes	& No	& No	& Yes	& Yes	\\
$\beta^{\psync}$	& No	& Yes	& Yes	& Yes	& Yes	\\
$\beta^{\sync}$		& No	& No	& Yes	& Yes	& Yes	\\ \bottomrule
\end{tabular}
\caption{An overview of all the possible bifurcation types in all non-trivial regions. }
\label{fig:bifoverview}
\end{table}
\end{remark}

\subsection{Classification in region 2}
In $\RR_2$ there is a maximum of two solutions. If two solutions exists, one is synchronized and one is unsynchronized. We split $\RR_2$ into sub-regions.

\subsubsection{The case $K_1 < 2$ and $K_2 < 2$}
By \cite[Theorem IV.1]{Achterhof2020}, part $1$ and $2$ no synchronized solution exists if $\beta^{\zero} \geq 0$. Furthermore, by Theorem \ref{thm:bifzero1}, bifurcation at zero occurs if $\beta^{\zero} = 0$. Hence the synchronized solution is the only solution if $\beta^{\zero} \geq 0$, and there exists a synchronized solution and an unsynchronized solution if $\beta^{\zero} < 0$. Figure \ref{c(1,3,2)} demonstrates how the bifurcation gives rise to more solutions in this region.

\subsubsection{The case $K_1 = K_2 = 2$}
By Theorem \ref{thm:bifzero1} bifurcation from zero occurs if $L_1 = 0$. A bifurcation diagram in this region is given in Figure \ref{b(2,2,3)}.

\subsubsection{The case $K_1 = 2, K_2 < 2$ or $K_1 < 2, K_2 = 2$}

By Theorem \ref{thm:bifzero3} bifurcation from zero occurs at $L_1 = 0$ or $L_2 = 0$. A bifurcation diagram in this region is given in Figure \ref{b(2,1,3)}.

\begin{table}
\centering
\vspace{0.1cm}
\begin{tabular}{l c l } \toprule
\textbf{Extra condition(s)}	& \textbf{ $\#$ solutions}  &\textbf{Classification}  \\ \midrule
$\beta^{\zero} \geq 0$	& 1	& 1 unsync \\ 
$\beta^{\zero} < 0$	& 2	& 1 unsync + 1 sync\\ \bottomrule
\end{tabular}
\caption{Classification in the region $\RR_2$.}\label{fig:tableR2}
\end{table}

This allows us to refine the region as can be seen in Table \ref{fig:tableR2}.

\subsection{Classification in regions 3 to 5}
By Lemma \ref{lem:nobif}, bifurcation from zero is not possible in these regions. Also, by the geometry of the level curves there are precisely two solutions in the whole region. By Lemma \ref{lem:solempty} bifurcation from the synchronized solution does not occur.

Furthermore, in $\RR_4$ and $\RR_5$ we have by Theorem \ref{thm:epop} that the synchronized solution occurs at a pop-up solution (at $L_1 = 0)$ when $L_1$ is varied. An example of a bifurcation diagram is given in Figure \ref{b(2,3,_,5)}. In addition, in Table \ref{fig:tableR3R4R5} a full classification in $\RR_3$, $\RR_4$ and $\RR_{5}$ is given.

\begin{table}[!ht]
\centering
\vspace{0.1cm}
\begin{tabular}{l c l} \toprule
\textbf{Extra condition(s)}	& \textbf{ $\#$ solutions}  &\textbf{Classification}  \\ \midrule
$\emptyset$	& 2	& 1 unsync + 1 sync\\ \bottomrule
\end{tabular}
\caption{Classification in the region $\RR_3, \RR_4, \RR_5$.}\label{fig:tableR3R4R5}
\end{table}

\subsection{Classification in regions 6 and 7}
By Theorem \ref{lem:nobif} bifurcation from zero is not possible in these regions, but pop-up bifurcation is possible. 

By Lemma \ref{lem:nobif} bifurcation from zero is not possible in these regions. Furthermore, by Lemma \ref{lem:solnonempty} bifurcation from the synchronized solution is possible. In other words: the solution boundary $\beta^{\zero}=0$ does not exist, but $\beta^{\sync} =0$ does. By Lemma \ref{lem:uniqsol} the solution boundary is unique and by Theorem \ref{thm:opsign} the asymptotes of $\beta^{\sync}=0$ are given. To understand the behavior of the solutions in this region we give in Figure \ref{c(25,-2,1)} a bifurcation diagram. We observe that a pop-up solution occurs at the bifurcation point. In Table \ref{fig:tableR6} the the full classification in $\RR_{6}$ and $\RR_{7}$ is given.

\begin{figure*}

\centering
\includegraphics[width=\linewidth]{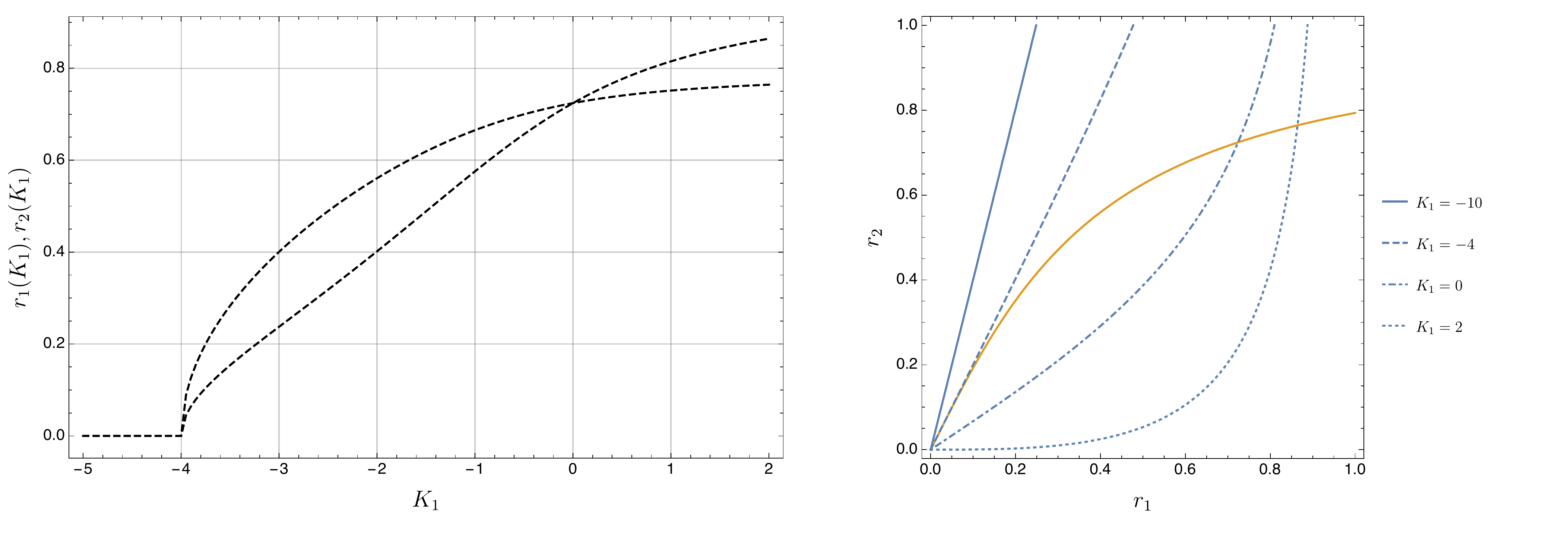} 
 \caption{Left: plot of $K_1 \mapsto r_1(K_1)$ and $K_1 \mapsto r_2(K_1)$ in $\RR_2$ with $K_2 = 1, L_1 = 3 $ and $L_2 = 2$. Bifurcation from zero occurs at $K_1^{\zero} = -4$. In addition, $K_1^{\sym} = 0$. Right: plot of the level curves $\Gamma_1$ and $\Gamma_2$ with the same interaction strengths $K_2, L_1, L_2$ and for different choices of $K_1$.}\label{c(1,3,2)}
\vspace{2em}

\begin{subfigure}[t]{.45\linewidth}
\includegraphics[width=\linewidth]{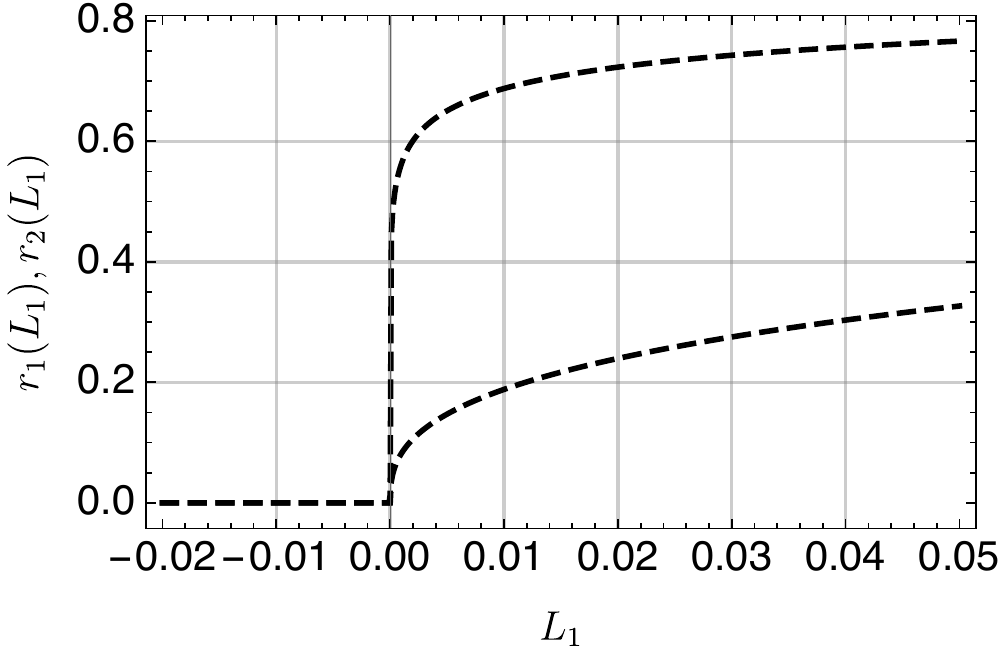} 
\caption{Plot of $L_1 \mapsto r_1(L_1)$ and $L_1 \mapsto r_2(L_1)$ in $\RR_1$, with $K_2 = 2, K_2 = 2 $ and $L_2 = 3$. Bifurcation from zero occurs at $L_1^{\zero} = 0$.}\label{b(2,2,3)}
\end{subfigure}%
\hspace{1em}
\begin{subfigure}[t]{.45\linewidth}
\includegraphics[width=\linewidth]{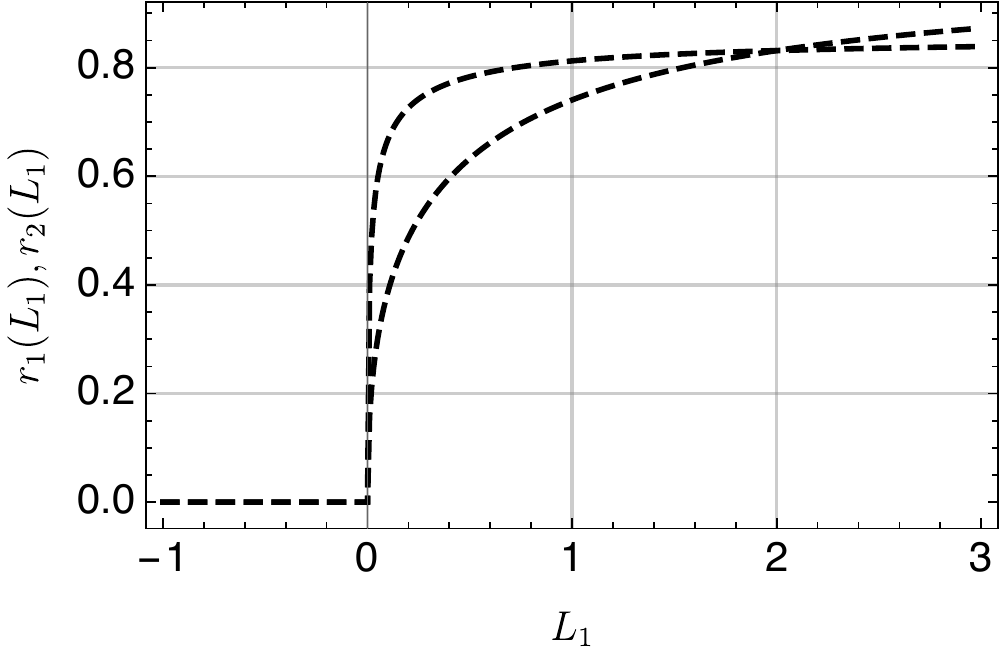} 
 \caption{Plot of $K_1 \mapsto r_1(K_1)$ and $K_2 \mapsto r_2(K_1)$ in $\RR_1$, with $K_2 = 2, K_2 = 1 $ and $L_2 = 3$. Bifurcation from zero occurs at $L_1^{\zero} = 0$. In addition, $K_1^{\sym} = 2$.}\label{b(2,1,3)}
\end{subfigure}
\caption{Two bifurcation diagrams in $\RR_2$.}
\vspace{2em}

\centering
\includegraphics[width=0.5\textwidth]{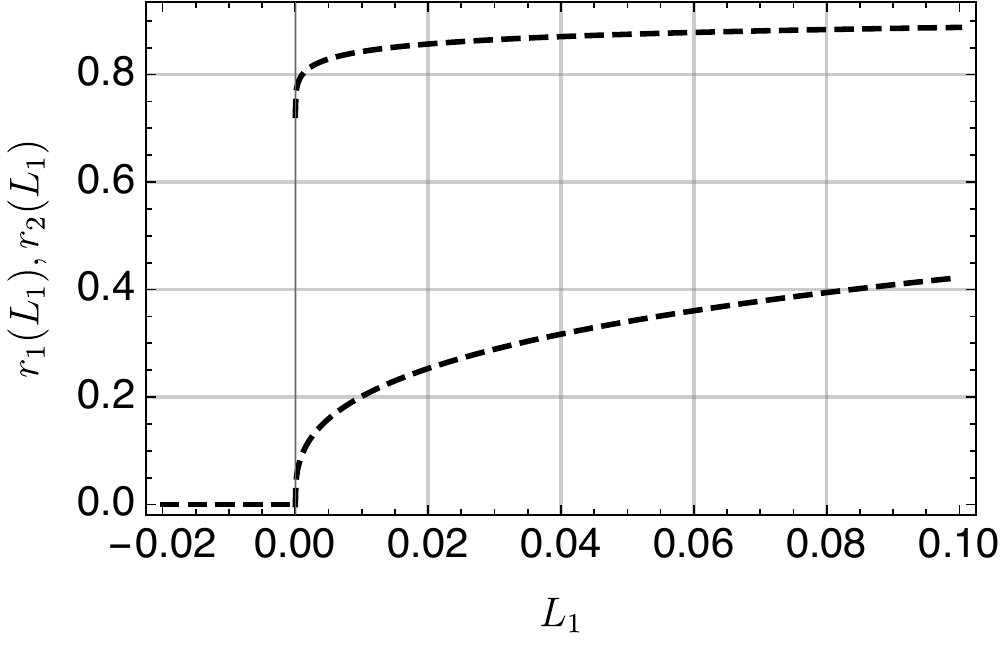}
\caption{Plot of $L_1 \mapsto (r_1(L_1), r_2(L_1))$ in $\RR_3$ (for $L_1 > 0$) when $K_1 = 2$, $K_2 = 3$, $L_2 = 5$ and $L_1$ is varied. Pop-up bifurcation occurs at $L_1^{\pu} = 0$ with $(r_1^{\pu}, r_2^{\pu}) = (0,0.724)$. }\label{b(2,3,_,5)}

\end{figure*}

\begin{table}[!ht]
\centering
\vspace{0.1cm}
\begin{tabular}{l c l} \toprule
\textbf{Extra condition(s)}		& \textbf{ $\#$ solutions}  &\textbf{Classification}  \\ \midrule
$\Gamma_1^{K_1, L_1} \cap \Gamma_2^{K_2, L_2} = \{ (0,0)\}$  	& 1	& 1 unsync \\ 
$\beta^{\sync} = 0$						& 2	& 1 unsync + 1 sync\\ 
$\beta^{\sync} \neq 0$					& 3	& 1 unsync + 2 sync\\ \bottomrule
\end{tabular}
\caption{Classification in the region $\RR_6$ and $\RR_7$.}\label{fig:tableR6}
\end{table}

\subsection{Classification in regions 8 and 9}
In region $\RR_8$ and $\RR_9$ bifurcation from zero and bifurcation from the unsynchronized solution is possible (see Theorem \ref{thm:bifzero1} and Lemma \ref{lem:solnonempty}) . In $\RR_8$, by \cite[Theorem IV.1]{Achterhof2020} part 5 and Theorem \ref{thm:bifzero1} if $\beta^{\zero} \geq 0$, then the unsynchronized solution is the only solution. By Lemma \ref{lem:solempty} bifurcation from a synchronized solution occurs. In Figure \ref{fig:region8} the two possible bifurcation diagrams in $\RR_8$ are sketched. In addition, in Table \ref{fig:tableR8R9} the full classification in $\RR_8$ and $\RR_9$ is given.

\begin{figure*}
\centering

\includegraphics[width=\linewidth]{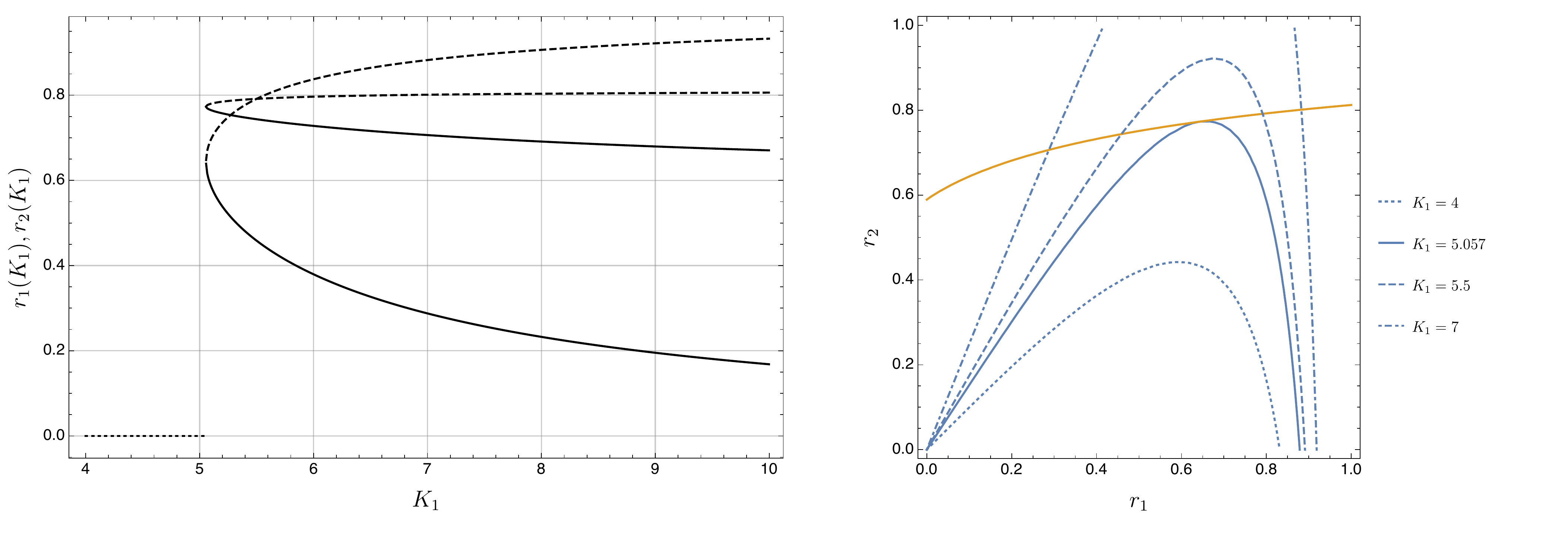} 
 \caption{Left: plot of $K_1 \mapsto (r_1(K_1), r_2(K_2))$ in $\RR_6$ when $K_2 = 2.5$, $L_1 = -2$, $L_2 = 1$ and $K_1$ is varied. Pop-up bifurcation occurs at $K_1^{\pu} = 5.057$, with $(r_1^{\pu}, r_2^{\pu}) = (0.6431,0.7719)$, and a symmetric solution appears at $K_1^{\sym} = 5.5$. Right: plot of the level curves $\Gamma_1$ and $\Gamma_2$ with the same interaction strengths $K_2, L_1, L_2$ and with $K_1$ varied.}\label{c(25,-2,1)}
\vspace{2em}

\begin{subfigure}{\linewidth}
\includegraphics[width=\linewidth]{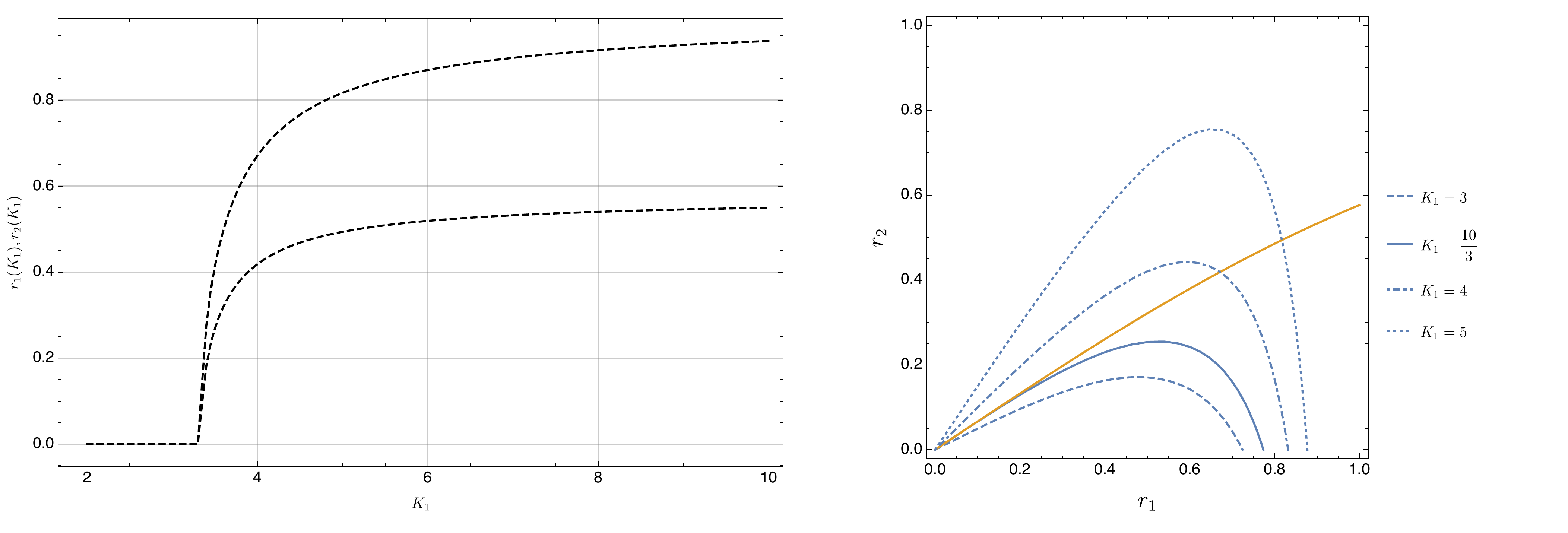} 
\caption{Left: plot of $K_1 \mapsto (r_1(K_1), r_2(K_2))$ in $\RR_8$ when $K_2 = -1$, $L_1 = -2$, $L_2 = 2$ and $K_1$ is varied. Furthermore, $K_1^{\pu} = K_1^{\zero} = \frac{10}{3}$ Right: plot of $\Gamma_1$ and $\Gamma_2$ with the same interaction strengths $K_2, L_1, L_2$ and $K_1$ is varied. }\label{c(-1,-2,2)}
\end{subfigure}

\begin{subfigure}{\linewidth}
\includegraphics[width=\linewidth]{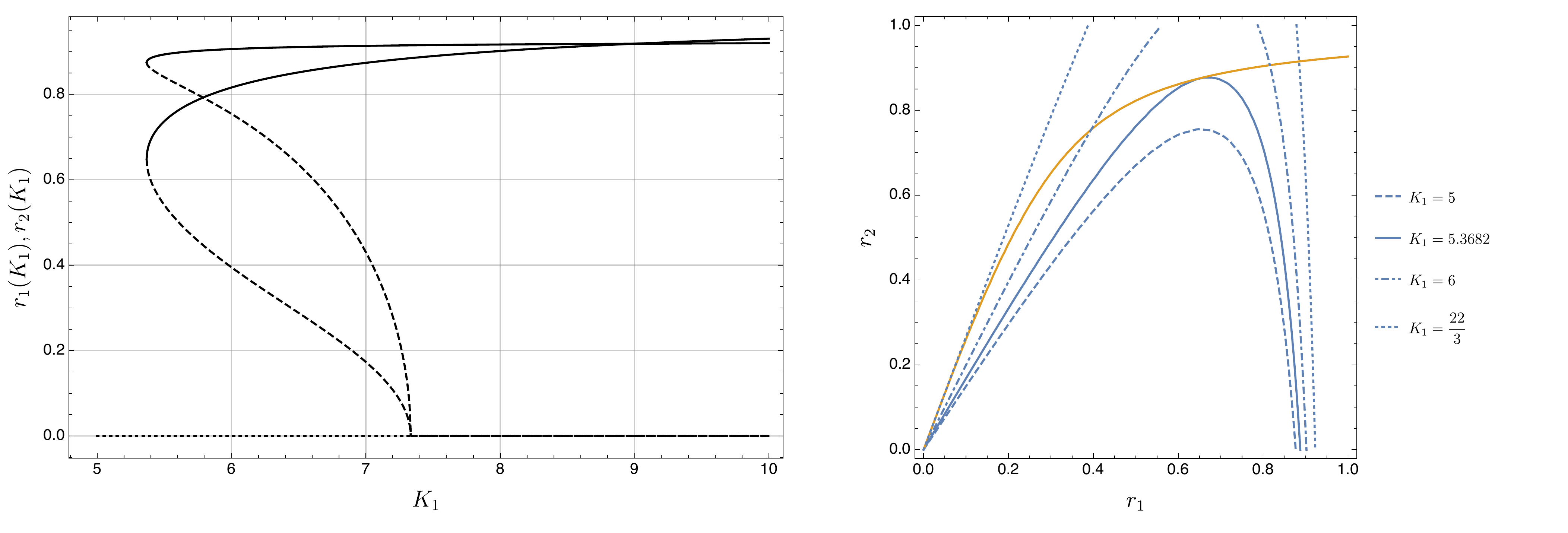}
\caption{Left: plot of $K_1 \mapsto (r_1(K_1), r_2(K_2))$ in $\RR_8$ when $K_2 = -1$, $L_1 = -2$, $L_2 = 8$ and $K_1$ is varied. Pop-up bifurcation occurs at $K_1^{\pu} = 5.3682$, with $(r_1^{\pu}, r_2^{\pu}) = (0.651,0.874)$, bifurcation from zero occurs at $K_1^{\zero} = \frac{22}{3}$, and a symmetric solution appears at $K_1^{\sym} = 9$. Right: plot of $\Gamma_1$ and $\Gamma_2$ with the same interaction strengths $K_2, L_1, L_2$ and with $K_1$ varied.}\label{c(-1,-2,8)}
\end{subfigure}

\caption{A numerical example of each of the two bifurcation diagrams in $\RR_8$.}
\label{fig:region8}
\end{figure*}

\begin{table}[!ht]
\centering
\vspace{0.1cm}
\begin{tabular}{l c l}\toprule
\textbf{Extra condition(s)}	& \textbf{ $\#$ solutions}  &\textbf{Classification} \\ \midrule
$\Gamma_1^{K_1, L_1} \cap \Gamma_2^{K_2, L_2} = \{ (0,0)\}$		& 1	& 1 unsync \\
$\beta^{\zero} > 0,~ \beta^{\sync} = 0$			& 2	& 1 unsync + 1 sync\\ 
$\beta^{\zero} < 0,~ \beta^{\sync} \neq 0$			& 2	& 1 unsync + 1 sync\\ 
$\beta^{\zero} > 0,~ \beta^{\sync} \neq 0$			& 3	& 1 unsync + 2 sync\\ \bottomrule
\end{tabular}
\caption{Classification in the region $\RR_8$ and $\RR_9$.}\label{fig:tableR8R9}
\end{table}

\subsection{Classification in region 10}
In this region, bifurcation from zero, bifurcation from a limit point and bifurcation from a synchronized solution can occur (see Theorem \ref{thm:bifzero1}, Theorem \ref{thm:biflim}, Lemma\ref{lem:solnonempty}). By \cite[Theorem IV.1]{Achterhof2020} part $3$ and $4$, the unsynchronized solution is the only solution if $\beta^{\zero} \leq 0$. Furthermore, by Theorem \ref{lem:nobf} the solution boundary $\beta^{\sync} = 0$ splits into two solution boundaries, namely $\beta^{\sync}_{1} = 0$ and $\beta^{\sync}_{2} = 0$. In Figure \ref{fig:bdr10} the possible bifurcation diagrams in $\RR_{10}$ are given. In addition, in Table \ref{fig:tableR10} the full classification in $\RR_{10}$ is given. 

\begin{figure*}
\centering
\begin{subfigure}{\linewidth}
\includegraphics[width=\linewidth]{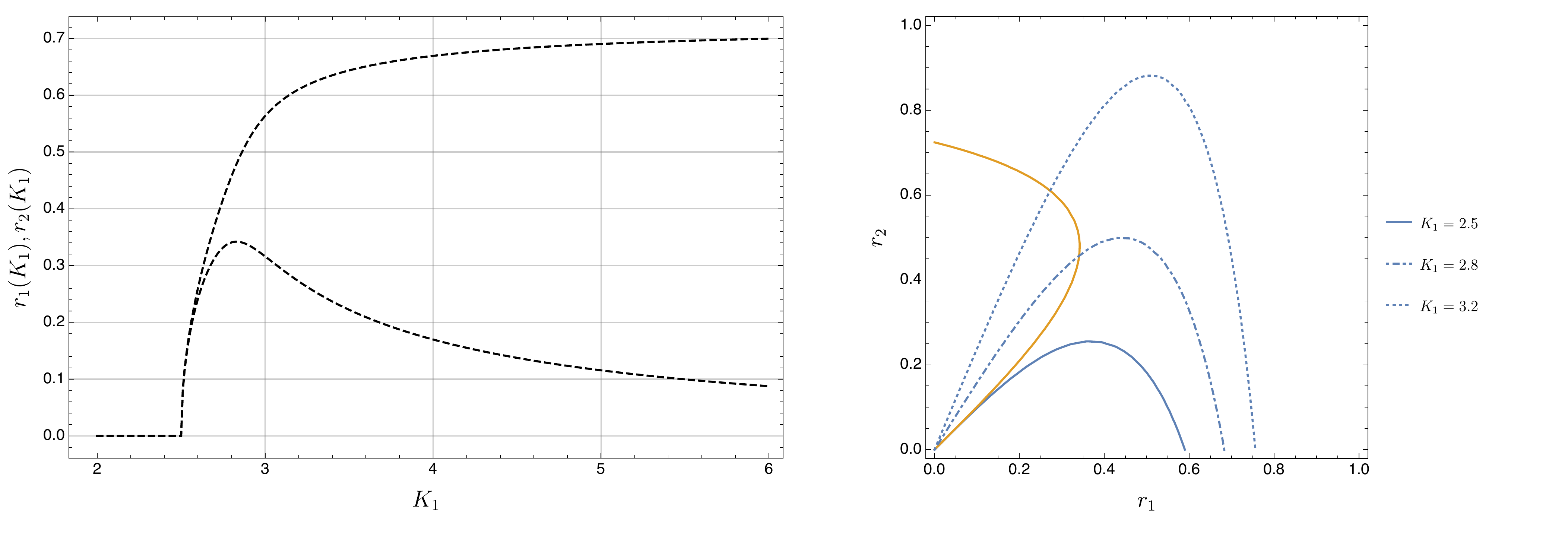} 
\caption{Left: plot of $K_1 \mapsto (r_1(K_1), r_2(K_2))$ in $\RR_{10}$ when $K_2 = 3$, $L_1 = -0.5$, $L_2 = -1$ and $K_1$ is varied. Bifurcation from zero occurs at $K_1^{\zero} = \frac{5}{2}$. Right: plot of $\Gamma_1$ and $\Gamma_2$ with the same interaction strengths $K_2, L_1, L_2$ and with $K_1$ varied. }\label{c(3,-0.5,-1)}
\end{subfigure}

\begin{subfigure}{\linewidth}
\includegraphics[width=\linewidth]{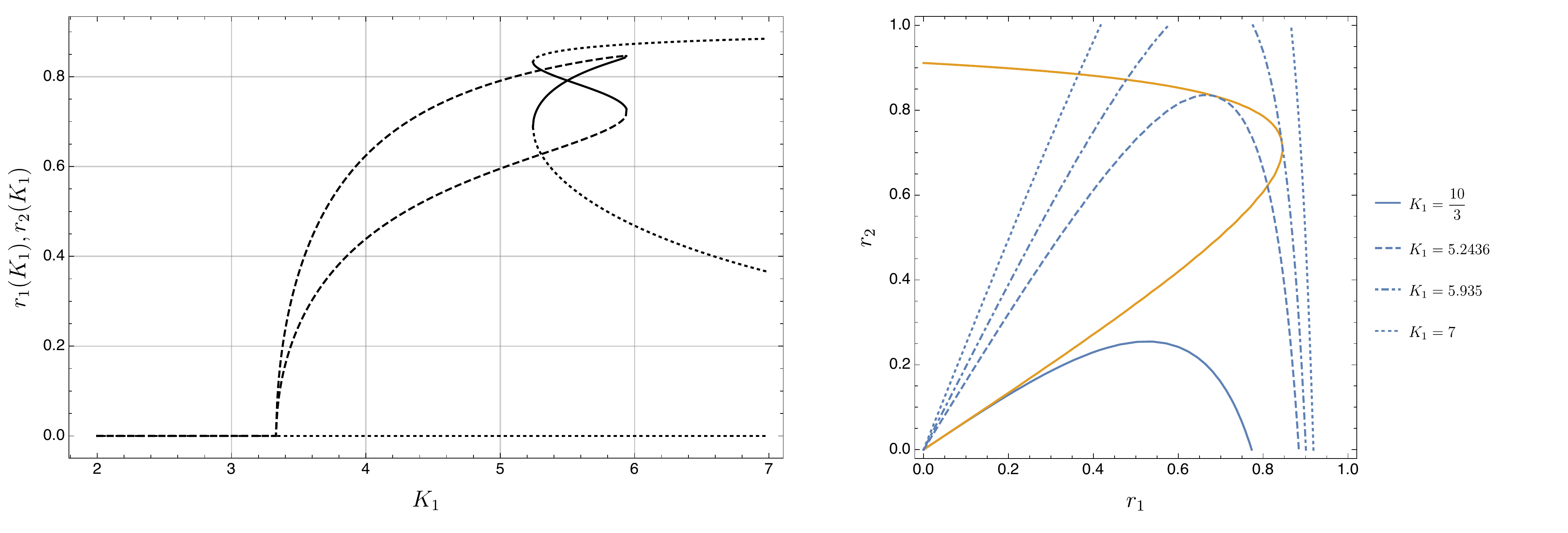} 
 \caption{Left: plot of $K_1 \mapsto (r_1(K_1), r_2(K_2))$ in $\RR_{10}$ when $K_2 = 6.5$, $L_1 = -2$, $L_2 = -3$ and $K_1$ is varied. Pop-up bifurcation occurs at $K_1^{\pu} = 5.244$, with $(r_1^{\pu}, r_2^{\pu}) = (0.685, 0.832)$, pop-down bifurcation occurs at $K_1^{\pd} = 5.935$, with $(r_1^{\pd}, r_2^{\pd}) = (0.846, 0.721)$, bifurcation from zero occurs at $K_1^{\zero} = \frac{10}{3}$, and a symmetric solution appears at $K_1^{\sym} = 5.5$. Right: plot of $\Gamma_1$ and $\Gamma_2$ with the same interaction strengths $K_2, L_1, L_2$ and with $K_1$ varied.}\label{c(6.5,-2,-3)}
\end{subfigure}

\begin{subfigure}{\linewidth}
\includegraphics[width=\linewidth]{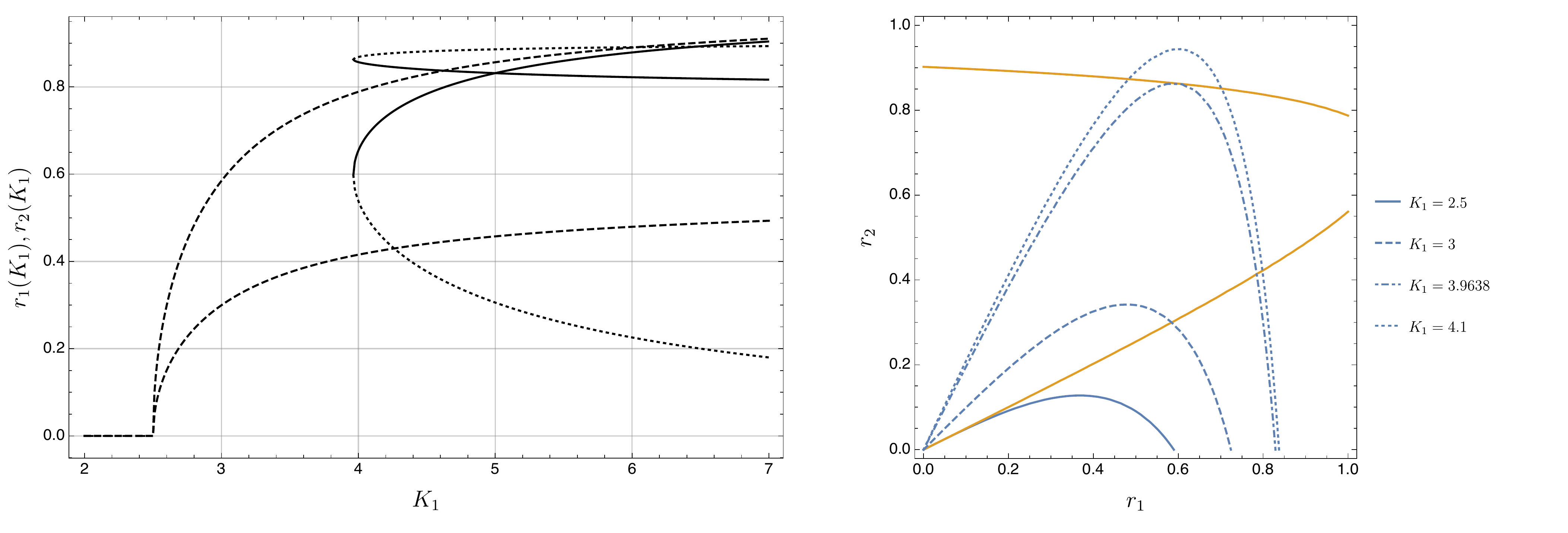} 
 \caption{Left: plot of $K_1 \mapsto (r_1(K_1), r_2(K_2))$ in $\RR_{10}$ when $K_2 = 6$, $L_1 = -1$, $L_2 = -2$ and $K_1$ is varied. Pop-up bifurcation occurs at $K_1^{\pu} = 3.964$, with $(r_1^{\pu}, r_2^{\pu}) = (0.599,0.862)$, bifurcation from zero occurs at $K_1^{\zero} = 2.5$ and a symmetric solution appears at $K_1^{\sym} = 5$. Right: plot of $\Gamma_1$ and $\Gamma_2$ with the same interaction strengths $K_2, L_1, L_2$ and with $K_1$ varied.}\label{c(6,-1,-2)}
\end{subfigure}

\caption{A numerical example of each of the three different bifurcation diagrams in $\RR_{10}$.}\label{fig:bdr10}
\end{figure*}

\begin{table}[!ht]
\centering
\vspace{0.1cm}
\begin{tabular}{l c l} \toprule
\textbf{Extra condition(s)}	& \textbf{ $\#$ solutions}  &\textbf{Classification}  \\ \midrule
$\beta^{\zero} \leq 0$									& 1	& 1 unsynchronized \\ 
$\text{Exterior}(\beta^{\sync}_{1} \oplus \beta^{\sync}_{2} = 0)$		& 2	& 1 unsync + 1 sync\\ 
$\beta^{\zero} > 0, \beta^{\sync}_{1} = 0$ and $\beta^{\sync}_{2} = 0$			& 2	& 1 unsync + 1 sync\\ 
$\beta^{\sync}_{1} = 0$ (strict) or $\beta^{\sync}_{2} = 0$				& 3 & 1 unsync + 2 sync\\ 
$\text{Interior}(\beta^{\sync}_{1} \oplus \beta^{\sync}_{2} = 0)$		& 4 & 1 unsync + 3 sync\\ \bottomrule
\end{tabular}
\caption{Classification in the region $\RR_{10}$.}\label{fig:tableR10}
\end{table}

\section{Regional phase diagrams }
\label{sec:phasediagrams}
From the refined classification of the previous section we can plot numerical examples of the phase diagrams that occur in various regions. We have done this for selected regions in Figure \ref{fig:phasediagrams1} and Figure \ref{fig:phasediagrams2}. Due to the difficulty of visualizing a phase diagram depending on four parameters we restrict ourselves to plotting slices of the phase space in which two of the parameters remain fixed. More specifically, we will fix either $K_1$ and $K_2$ or $L_1$ and $L_2$. We choose this representation due to the complexity of the asymptotes (see Theorem \ref{def:asymp}). The asymptotes in Theorem \ref{def:asymp} occur as points (instead of lines) in the chosen representations.

\begin{remark}[Coloring of the solution boundaries]
We distinguish the regions with a different possible number of solutions with different colors. The coloring is now as follows:

\begin{enumerate}
\item In the red area precisely one solution exists, namely, the unsynchronized solution. The red solution boundary always corresponds with $\beta^{\zero} = 0$.
\item In the green area there exist precisely two solutions. The green solution boundary corresponds with either $\beta^{\sync} = 0$ or $\beta^{\zero} = 0$.
\item In the blue area there exist precisely three solutions. The blue solution boundary only occurs if $L_1 < 0$ and $L_2 < 0$. This solution boundary always corresponds with $\beta^{\sync}_{1} = 0$ and $\beta^{\sync}_{2} = 0$.
\item In the yellow area there exist precisely four solutions.
\end{enumerate}
\end{remark}

\begin{remark}[Relation phase- and bifurcation diagram]
A phase diagram can be related to a bifurcation diagram. Examples are the bifurcation diagram Figure \ref{c(6.5,-2,-3)} and the phase diagram Figure \ref{fig:l1-2andl2-3}. For both diagrams $L_2 = -2$ and $L_2 = -3$. If we fix $K_2 = 6.5$ in Figure \ref{fig:l1-2andl2-3} and let $K_1$ vary then bifurcation points of Figure \ref{c(6.5,-2,-3)} (i.e., $K_1^{\zero}, K_1^{\pu}, K_1^{\sym}$ and $K_1^{\pd}$) occur at the solution boundaries of Figure \ref{fig:l1-2andl2-3}. To be more specific, by following the horizontal line (from left to right) corresponding to Figure \ref{fig:l1-2andl2-3}, we first cross the red solution boundary at $K_1^{\zero} = \tfrac{10}{3}$, then cross the blue solution boundary at $K_1^{\pu} = 5.244$, next arrive in the yellow area and cross the dashed line, which corresponds to the symmetric solution at $K_1^{\sym} = 5.5$, and finally arrive at the second blue solution boundary, which corresponds to $K_1^{\pd} = 5.935$.
\end{remark}

\begin{figure*}
\centering
\begin{subfigure}{.4\linewidth}
  \centering
  \includegraphics[width=\linewidth]{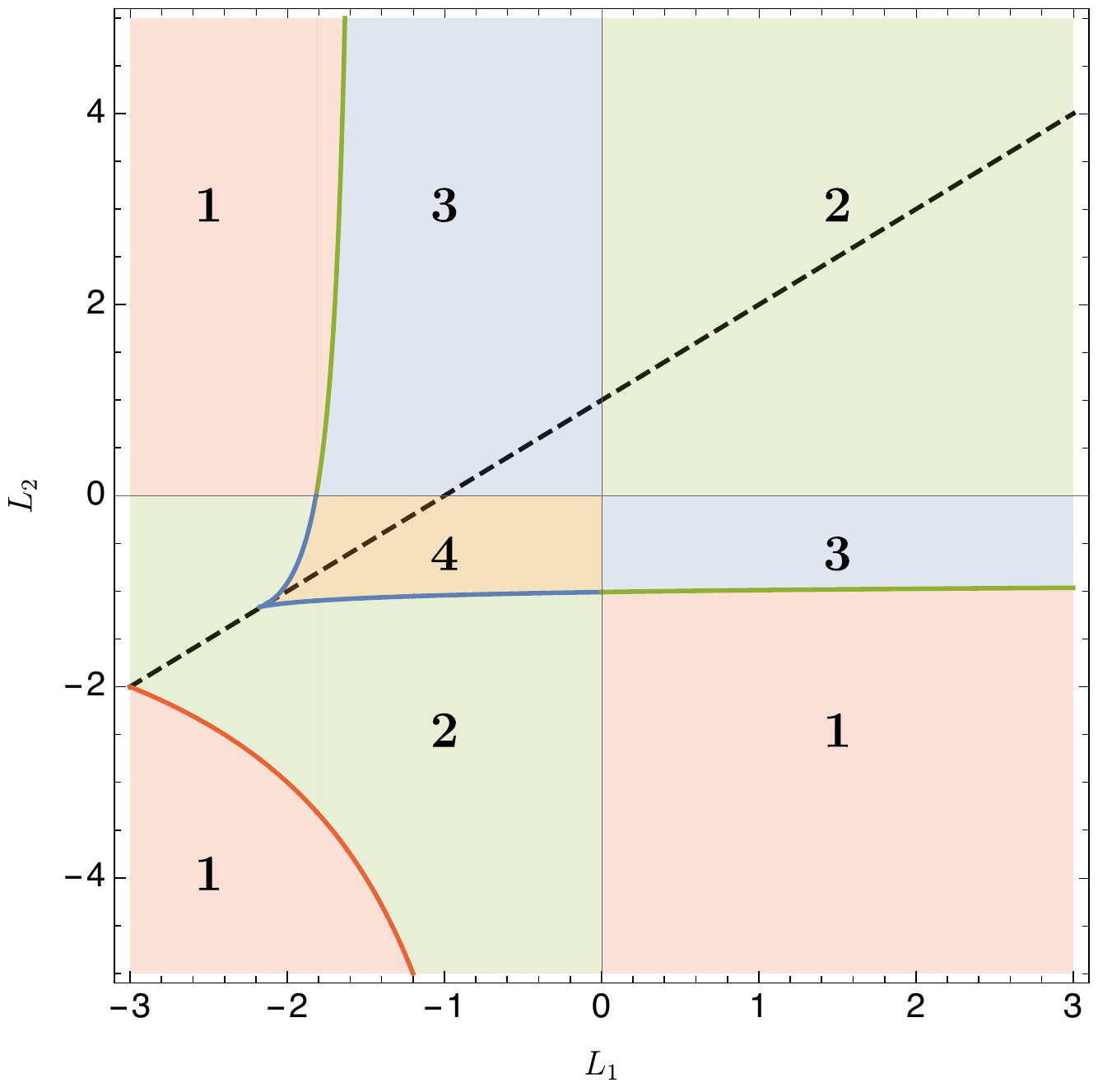}
  \caption{$K_1 = 5, K_2 = 4$.}
  \vspace{1em}
\end{subfigure}%
\begin{subfigure}{.4\linewidth}
  \centering
  \includegraphics[width=\linewidth]{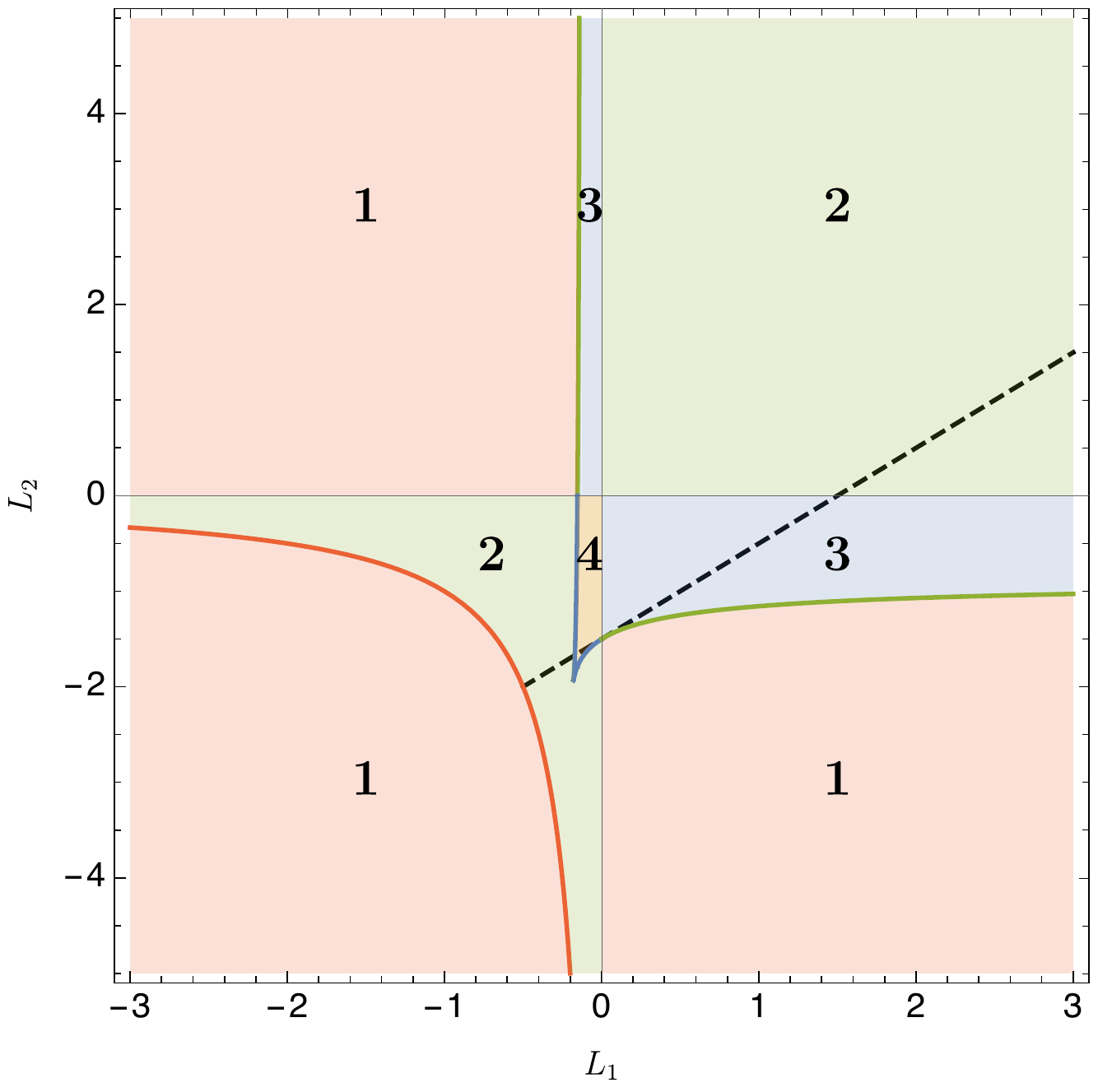}
  \caption{$K_1 = 2.5, K_2 = 4$.}
  \vspace{1em}
\end{subfigure}

\begin{subfigure}{.4\linewidth}
  \centering
  \includegraphics[width=\linewidth]{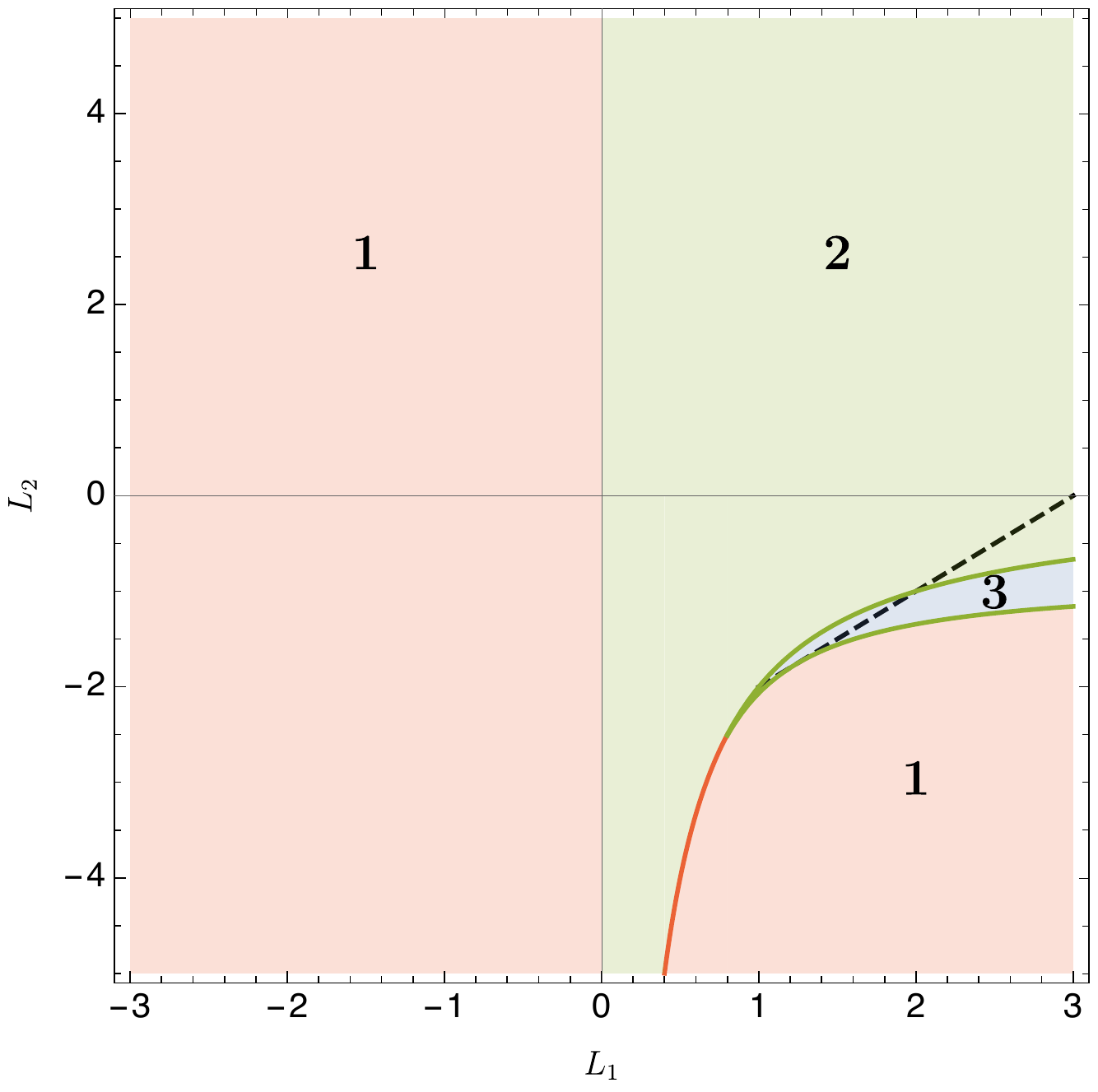}
  \caption{$K_1 = 1, K_2 = 4$.}
\end{subfigure}%
\begin{subfigure}{.4\linewidth}
  \centering
  \includegraphics[width=\linewidth]{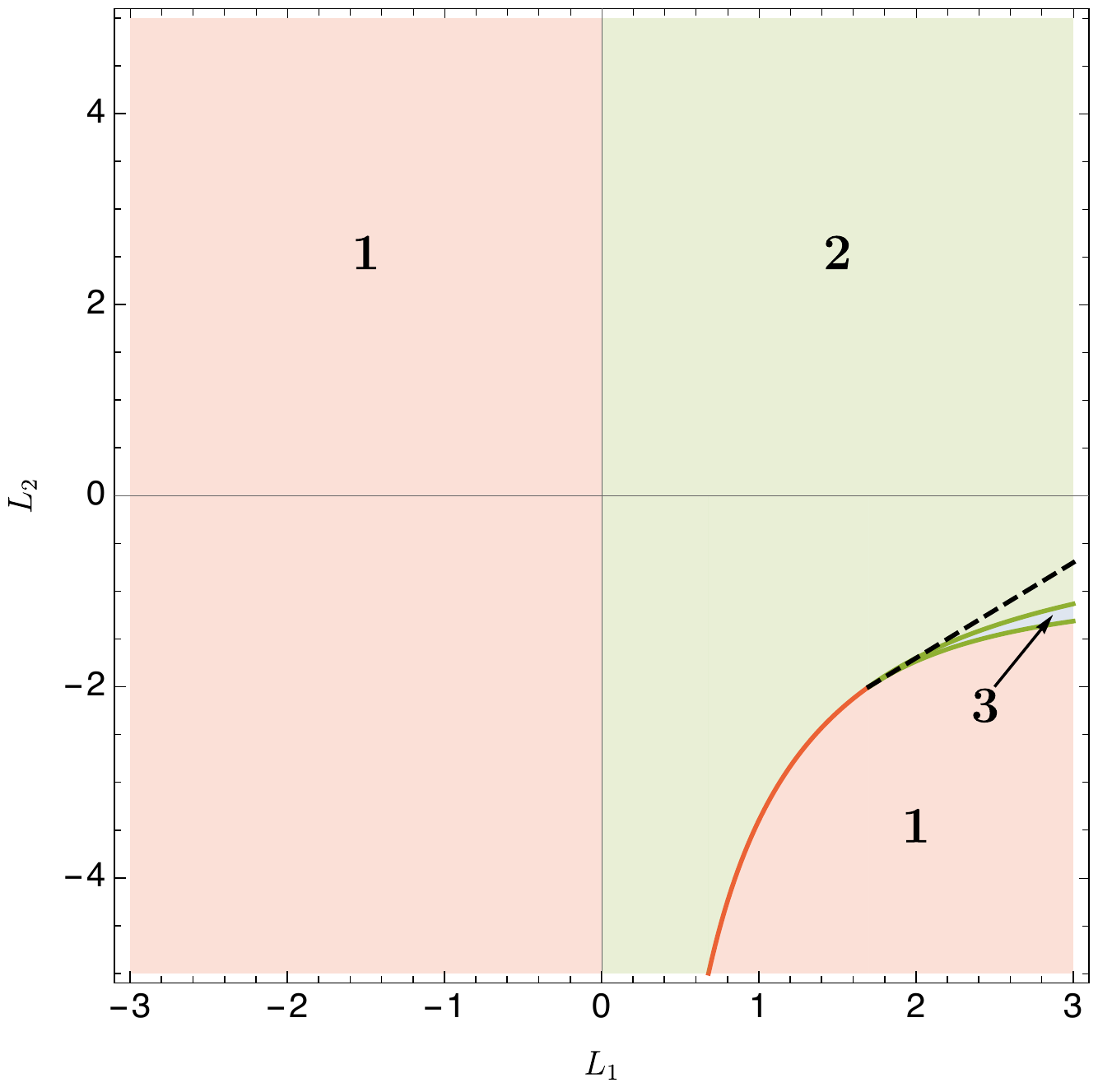}
  \caption{$K_1 = 0.3, K_2 = 4$.}
\end{subfigure}

\begin{subfigure}{.4\linewidth}
  \centering
  \includegraphics[width=\linewidth]{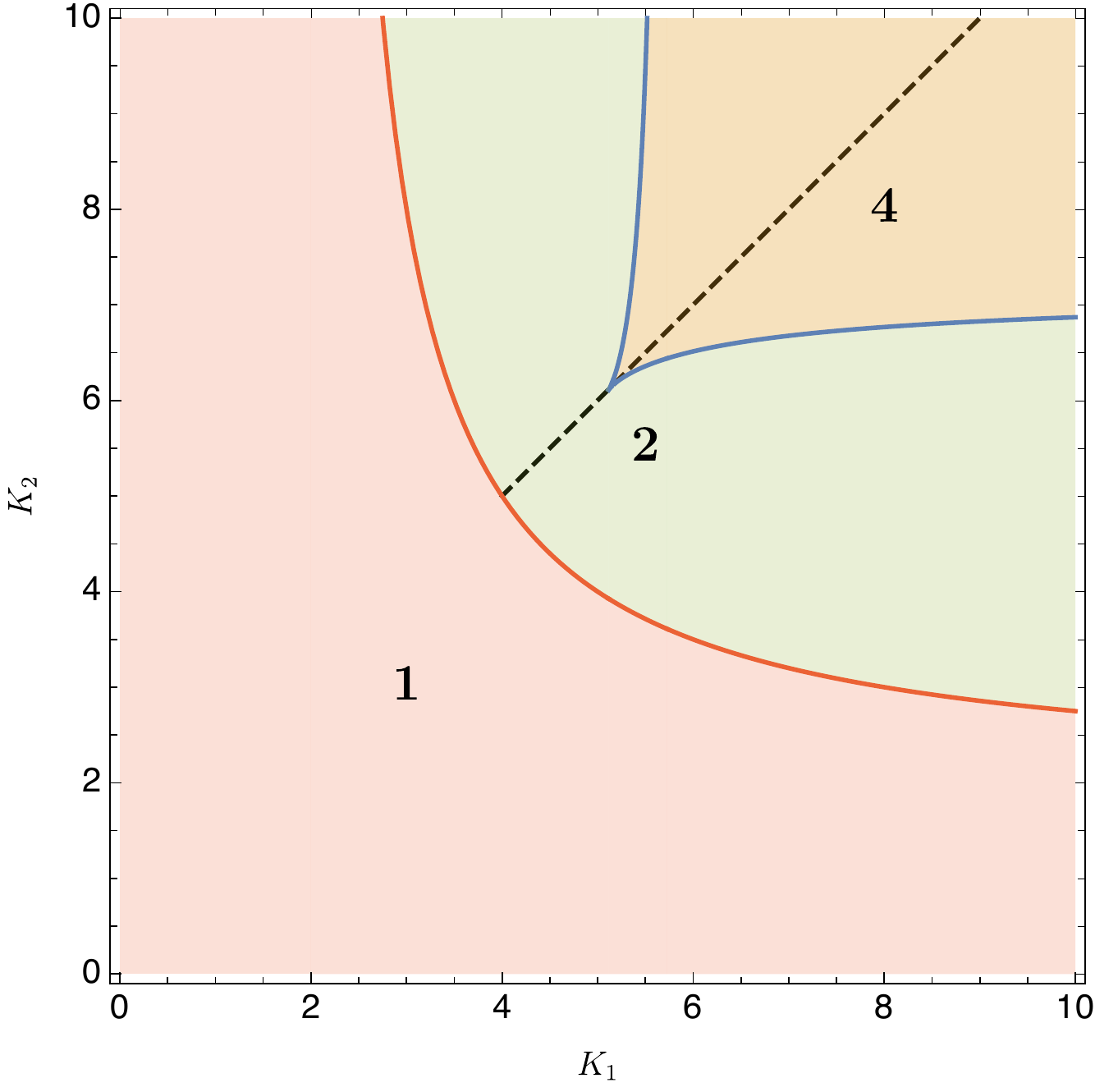}
  \caption{$L_1 = -2, L_2 = -3$.}
  \vspace{1em}
  \label{fig:l1-2andl2-3}
\end{subfigure}
\begin{subfigure}{.4\linewidth}
  \centering
  \includegraphics[width=\linewidth]{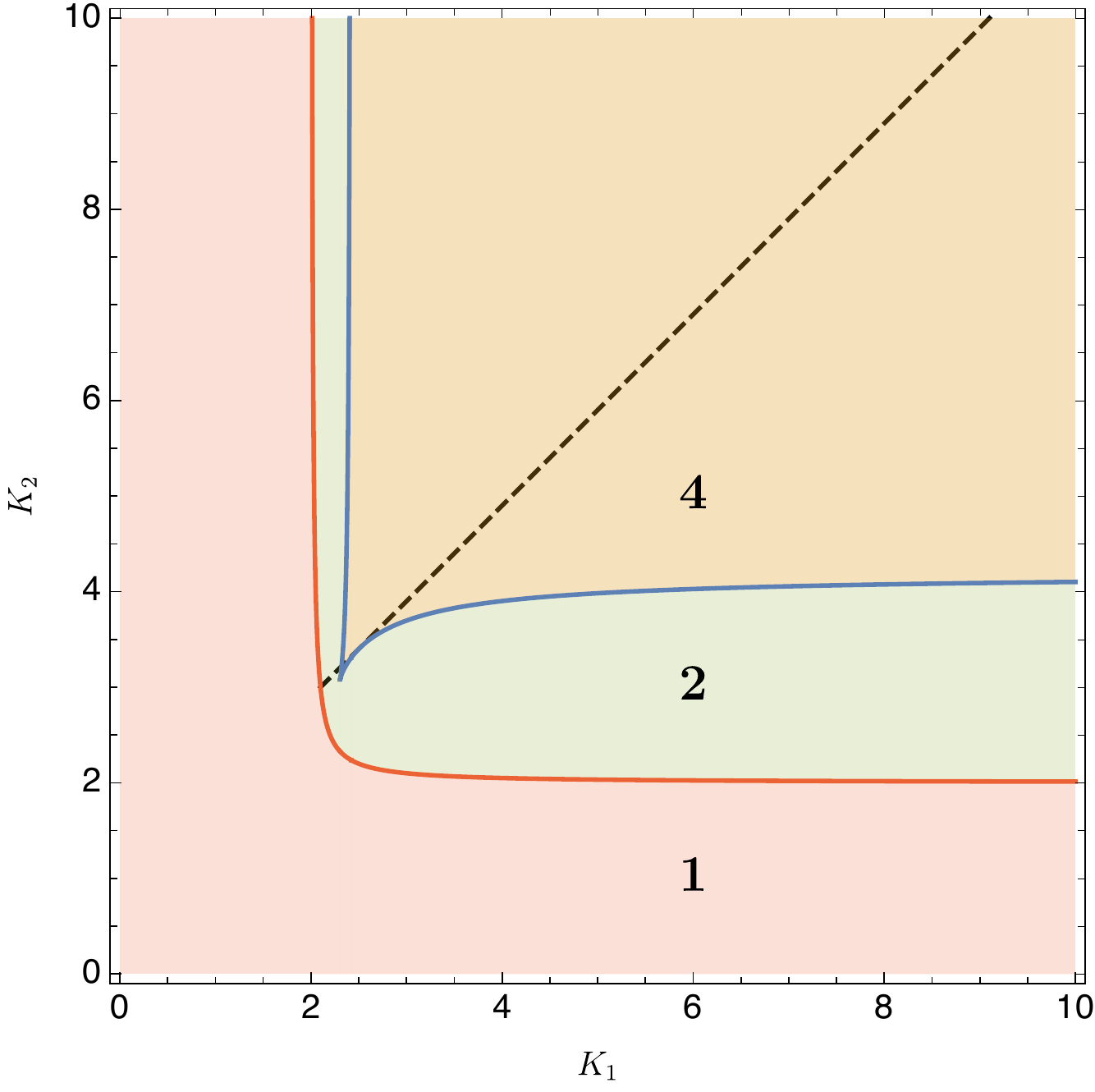}
  \caption{$L_1 = -0.1, L_2 = -1$.}
  \vspace{1em}
\end{subfigure}

\caption{In (a)-(d) plots of the solution regions with $K_2 = 4$, $L_1, L_2$ varied and for decreasing values of $K_1$ are given. Furthermore, in (e)-(f) for two choices of $L_1, L_2$ the solution regions are given when $K_1, K_2$ are varied.}
\label{fig:phasediagrams1}
\end{figure*}

\begin{figure*}
\centering
\begin{subfigure}{.4\linewidth}
  \centering
  \includegraphics[width=\linewidth]{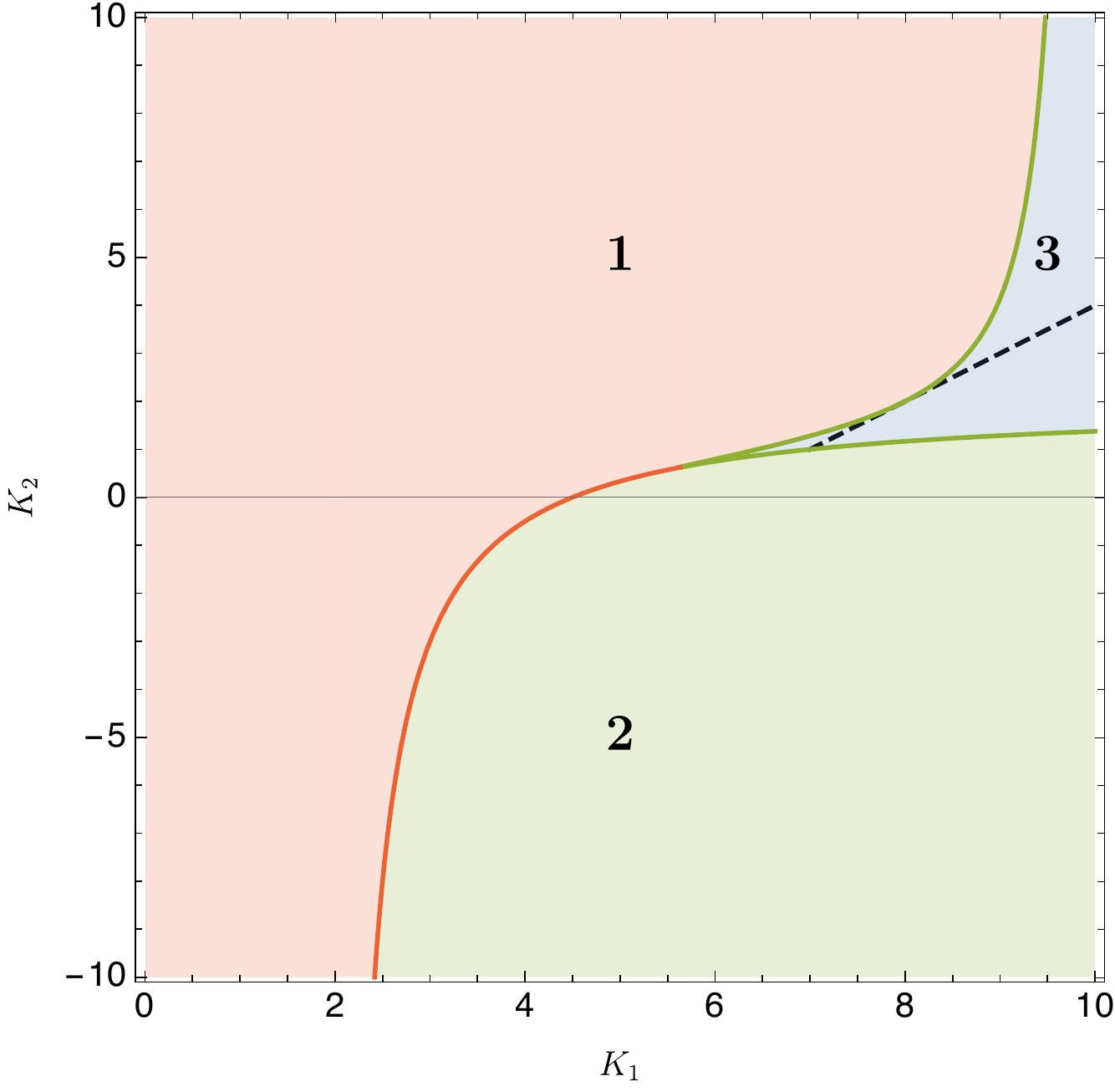}
  \caption{$L_1 = -5, L_2 = 1$.}
  \vspace{1em}
\end{subfigure}
\begin{subfigure}{.4\linewidth}
  \centering
  \includegraphics[width=\linewidth]{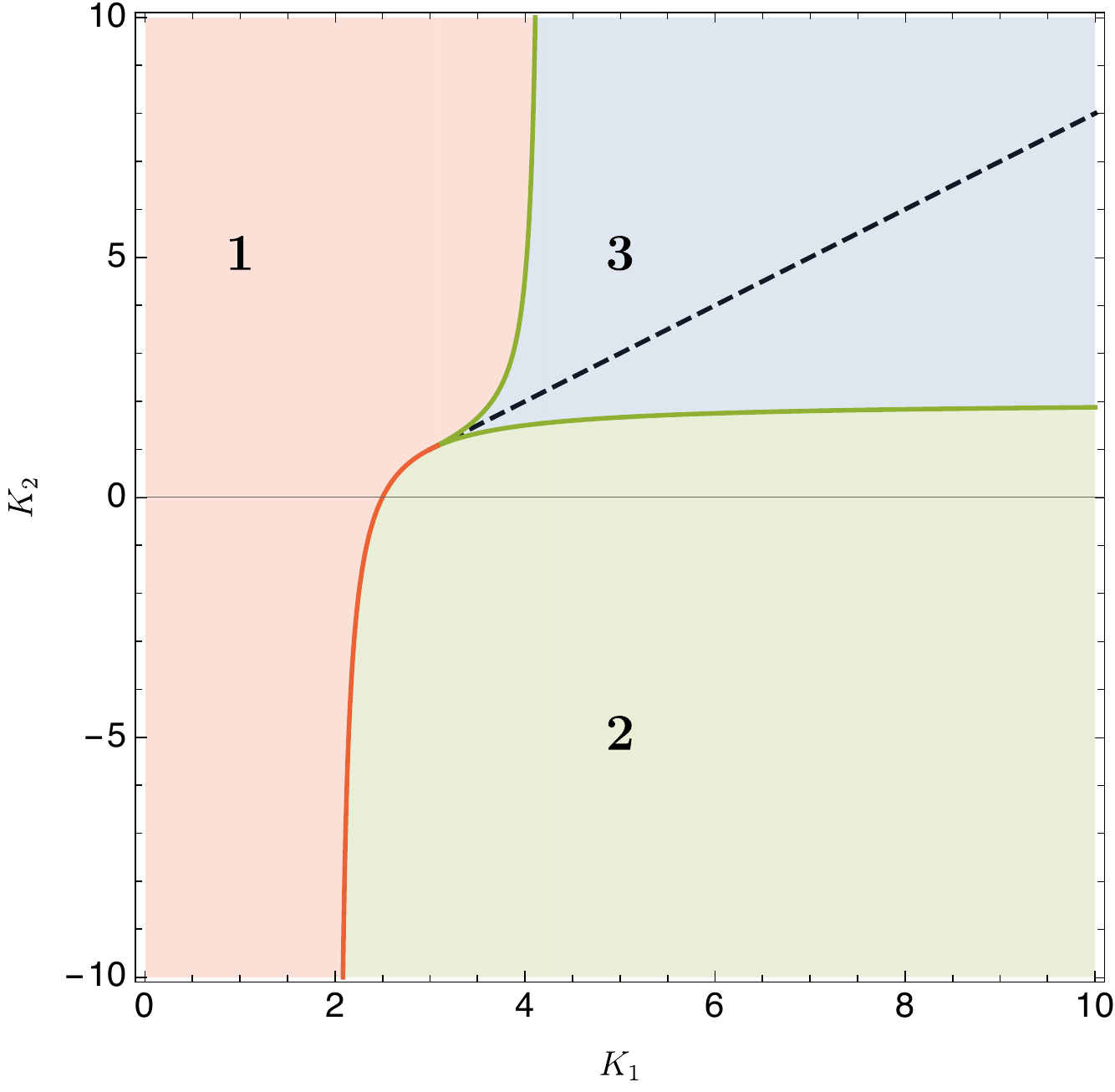}
  \caption{$L_1 = -1, L_2 = 1$.}
  \vspace{1em}
\end{subfigure}

\begin{subfigure}{.4\linewidth}
  \centering
  \includegraphics[width=\linewidth]{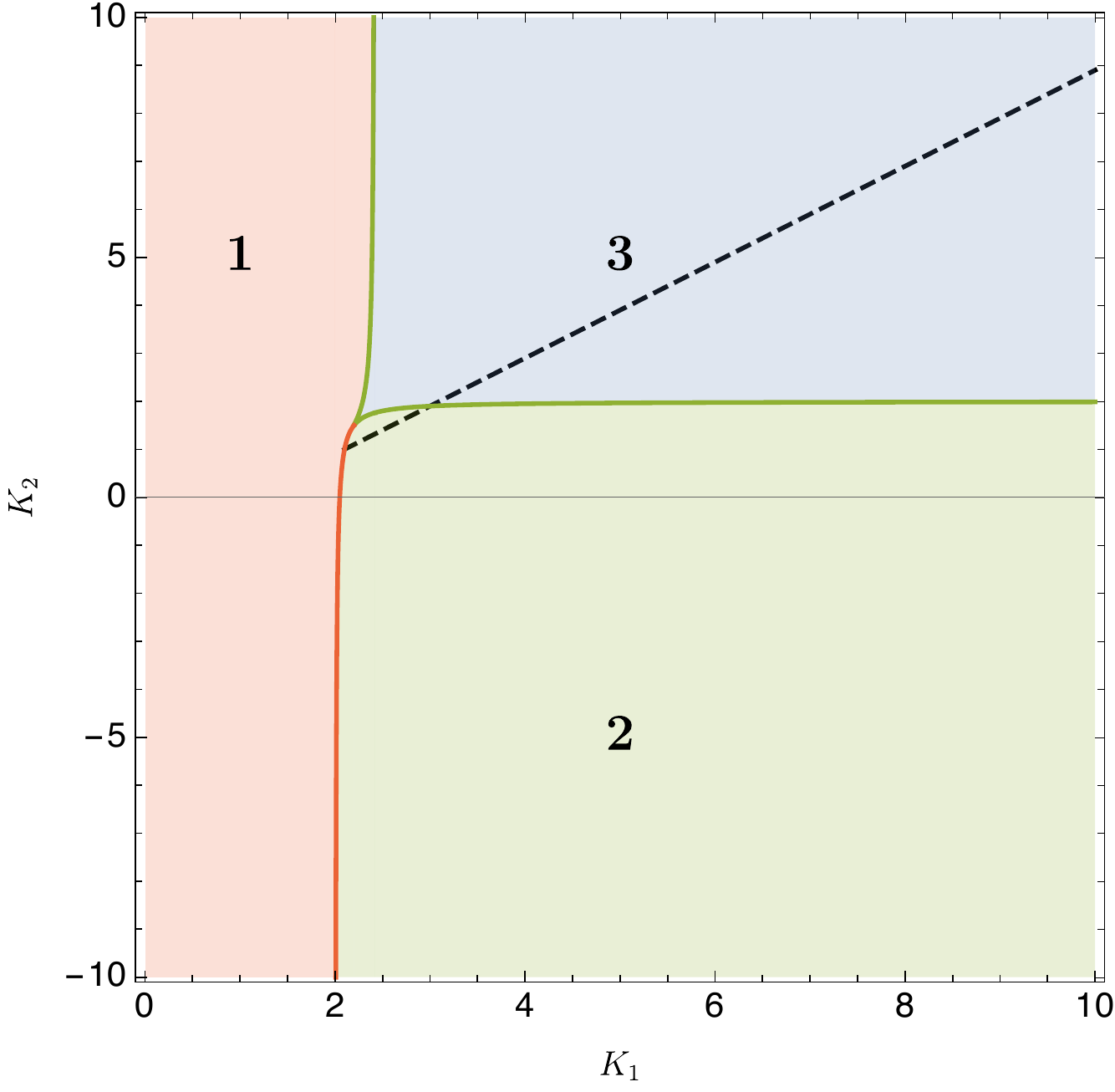}
  \caption{$L_1 = -0.1, L_2 = 1$.}
  \vspace{1em}
\end{subfigure}
\begin{subfigure}{.4\linewidth}
  \centering
  \includegraphics[width=\linewidth]{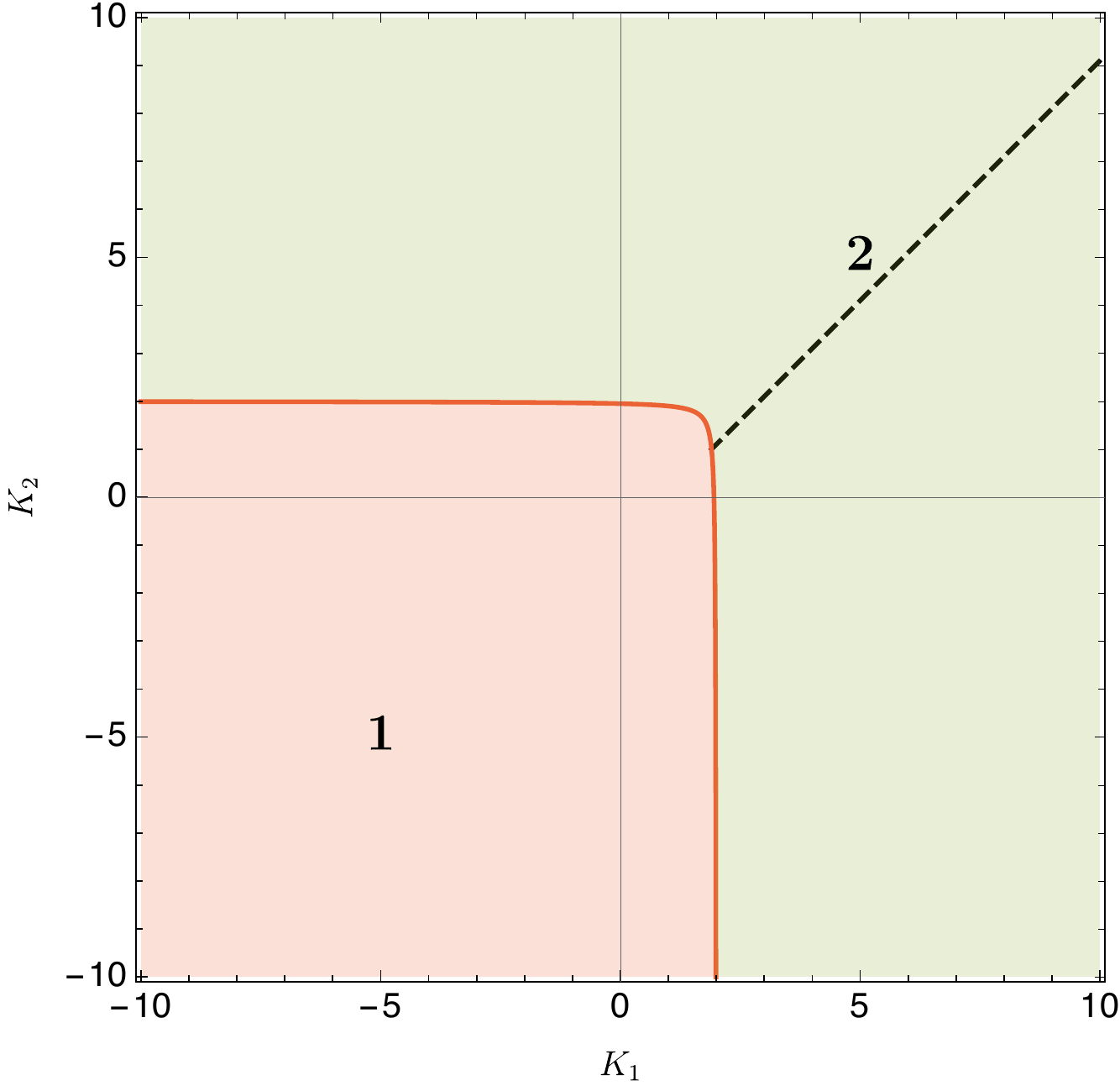}
  \caption{$L_1 = 0.1, L_2 = 1$.}
  \vspace{1em}
\end{subfigure}

\begin{subfigure}{.4\linewidth}
  \centering
  \includegraphics[width=\linewidth]{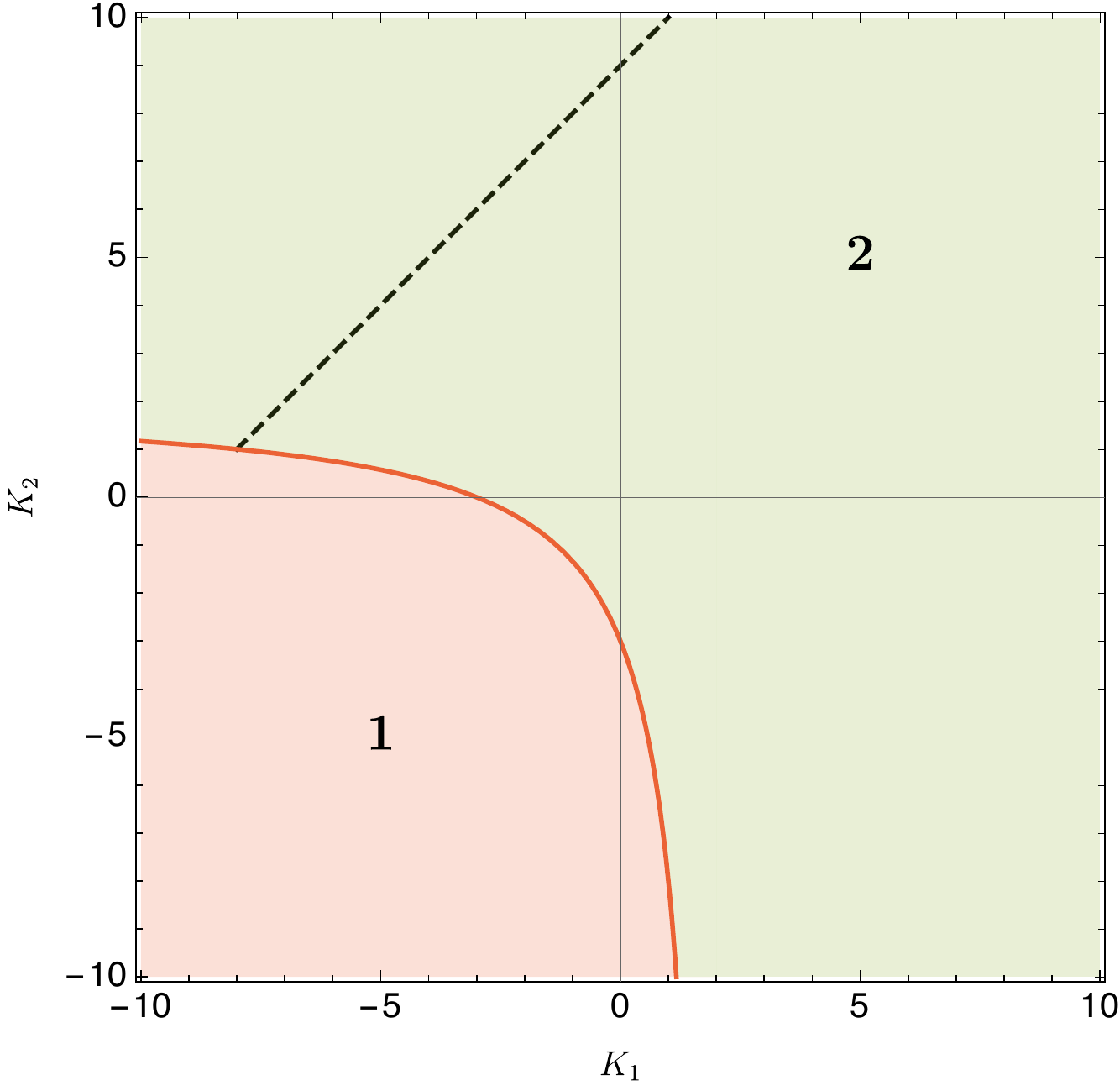}
  \caption{$L_1 = 10, L_2 = 1$.}
\end{subfigure}
\begin{subfigure}{.4\linewidth}
  \centering
  \includegraphics[width=\linewidth]{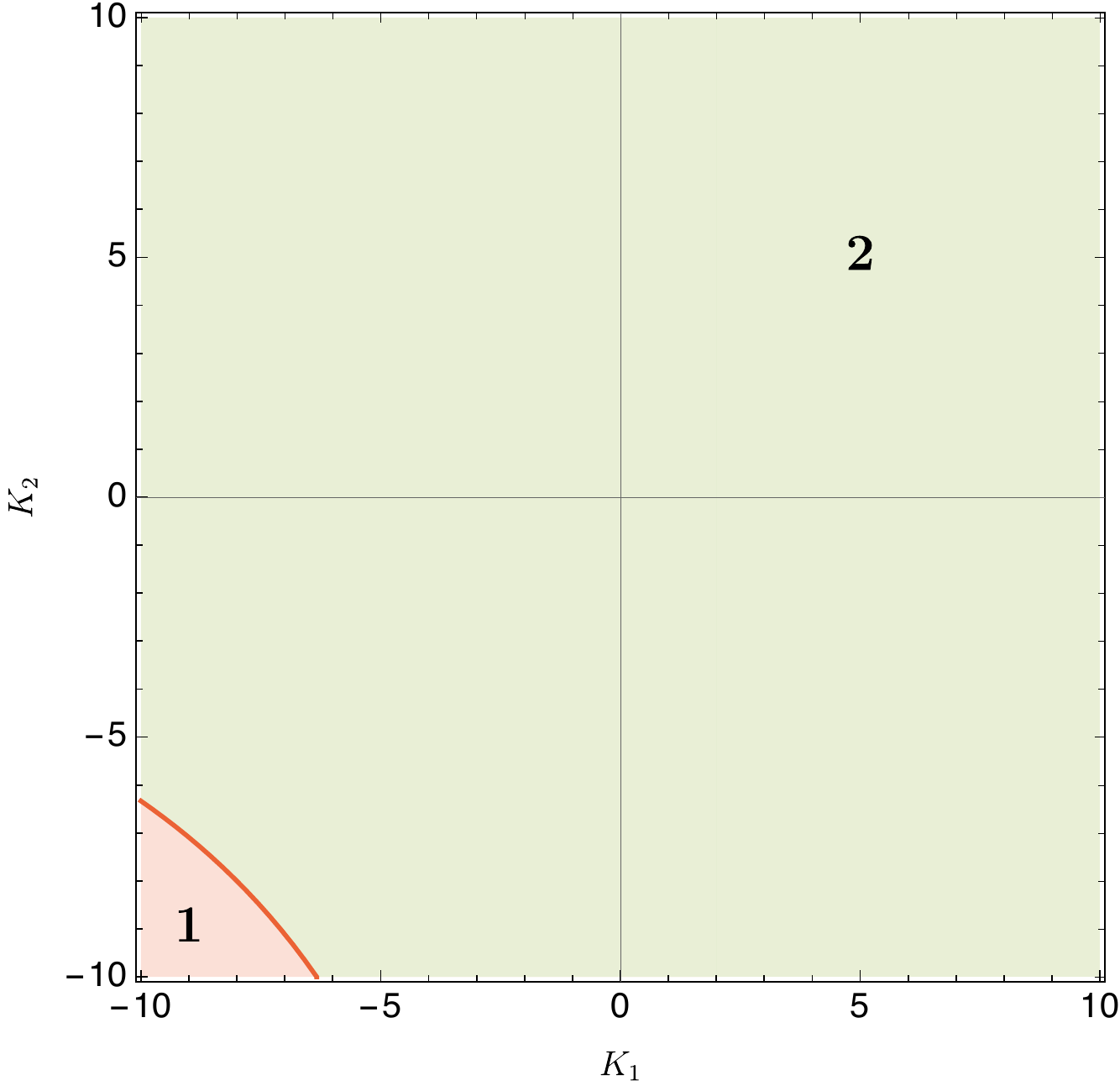}
  \caption{$L_1 = 100, L_2 = 1$.}
\end{subfigure}

\caption{Plot of the solution regions with $L_2 = 4$, $K_1, K_2$ varied and for increasing values of $L_1$.}
\label{fig:phasediagrams2}
\end{figure*}

\section{Conclusion}
We have introduced solution boundaries to partition each of the ten fundamental regions arising in the phase diagram of the two-community noisy Kuramoto model into subregions in which we precisely know the number of synchronized solutions. This is the first fundamental result of the two-community Kuramoto model with general interaction strengths. This phase diagram enables us to understand the response of the system when the interaction strengths change, e.g., when the interaction strength depend on time. This may give insight into the mechanics of the SCN by, for example, using single neuron SCN data to estimate the parameters of the two-community noisy Kuramoto model in various environmental conditions. These estimated parameters can in turn be used to make provable predictions regarding system mechanics. Research in this direction is ongoing. 

An interesting open problem is the stability of the synchronized solutions in the two-community noisy Kuramoto model. A stability analysis of the symmetrically synchronized solutions is a realistic starting point for the stability analysis, since in this case there are many similarities with the one-community noisy Kuramoto model.

\begin{acknowledgments}
The authors are grateful to F.\ den Hollander for guiding discussions and detailed comments.
\end{acknowledgments}

\section*{Data Availability}
Data sharing is not applicable to this article as no new data were created or analyzed in this study.

\appendix
\section{Proof of Theorem \ref{thm:opsign}}
\label{app:opsign}

\begin{proof}
The main idea is as follows. We consider the geometric configuration of a level curve in each of the values displayed in the left column of Table \ref{fig:astab1} and Table \ref{fig:astab2}. Then we argue when $\p \Gamma^{K_1,L_1}_1/\p r_1 = \p \Gamma^{K_2,L_2}_2/\p r_1$. The proof relies heavily on the geometry of the level curves. Therefore it is important to understand what the shape of the level curve is and how the level curves ``grow" as the interaction strengths are varied. In addition, it is important to understand the geometric configurations at the asymptotes (see Remark \ref{rem:interp}). We will prove the limits in Table \ref{fig:astab1}. Then the limits in Table \ref{fig:astab2} follow by a similar argument.

\begin{itemize}
\item Suppose that $K_1 \to K_1^a(L_1)$ or $L_1 \to L_1^a(K_1)$. Then the top of the level curve $\Gamma_1$ touches the line $[0,1] \times \{1\}$ (see Figure \ref{fig:Lnegstart}). Note that $\Gamma_2$ is strictly concave (since $L_2 > 0$) and therefore $\Gamma_1$ and $\Gamma_2$ can only intersect at the top of $\Gamma_1$. Furthermore, at the top of $\Gamma_1$ the derivative equals $\p \Gamma_1^{K_1^a, L_1}/\p r_1 = 0$. Therefore we require that either $K_2^* = \infty$ or $L_2^* = \infty$. 

\item Suppose that $K_1 \to K_1^b(K_2, L_1)$ or $L_1 \to L_1^b(K_1,K_2)$. Then the top of the parabola $\Gamma_1$ touches the vertical line drawn from the intersection point of $\Gamma_2$ with the line $\{0\} \times [0,1]$ (see Figure \ref{fig:Lnegb3}). Note that the intersection point of $\Gamma_2$ with the line $\{0\} \times [0,1]$ does not change when we change $L_2$. Since $\p \Gamma_2^{K_2, L_2}/\p r_1 > 0$ the intersection $\Gamma^{K_1, L_1}_1 \cap \Gamma^{K_2, L_2}_2$ is empty unless $L_2^* = 0$. This is true because $\p \Gamma_2^{K_2, 0}/\p r_1 = 0$. 

\item Suppose that $K_2 \to \infty$ or $L_2 \to \infty$, then 
\begin{equation}
\Gamma_2 \setminus \{ (0, 0) \} \to (0,1) \times \{1 \}, \label{eq:gamr2}
\end{equation}
point wise. This implies that $\p \Gamma^{\infty, L_2}/\p r_1 = 0$ and  $\p \Gamma^{K_2, \infty}/\p r_1 = 0$. Which means that if $\p \Gamma_1/\p r_1 = \p \Gamma_2/ \p r_1$, then this intersection occurs at the top of $\Gamma_1$. In addition, \eqref{eq:gamr2} requires that $r_2 = 1$, which implies that either $K_1^* = K_1^a$ or $L_1^* = L_1^a$.

\item Suppose that $L_2 \to 0$ (and $K_2 > 2$). Then $\Gamma_2$ has a non-trivial intersection with the axis $\{0\} \times [0,1]$. Furthermore $\p \Gamma^{K_2, 0}/ \p r_1 = 0$. Hence in order to have $\p \Gamma^{K_1,L_1}_1/\p r_1 = \p \Gamma^{K_2,L_2}_2/\p r_1$, we require that $K_1^* = K_1^b$ or $L_1^* = L_1^b$.
\end{itemize}
The limits in Table \ref{fig:astab2} follow by the same argument, where $K_1$, $K_2$ and $L_1$, $L_2$ are interchanged.
\end{proof}

\end{document}